%% file: main.tex
\documentclass[aps, prb, english, notitlepage, twocolumn, superscriptaddress, floatfix]{revtex4-2} 
\usepackage[utf8]{inputenc}
\usepackage{CJK}
\usepackage{graphicx}
\usepackage{dcolumn}
\usepackage{bm}
\usepackage{mathtools}
\usepackage{amssymb,amsmath,amsfonts,dsfont, amsthm}
\theoremstyle{definition}

\newtheorem{prop}{Proposition}

\usepackage{hyperref}
\usepackage{color}

\newcommand{\tr}{\operatorname{Tr}}

\def\bra#1{\mathinner{\langle{#1}|}}
\def\ket#1{\mathinner{|{#1}\rangle}}
\def\braa#1{\mathinner{\langle\!\langle{#1}|}}
\def\kett#1{\mathinner{|{#1}\rangle\!\rangle}}

\newcommand{\proj}[1]{\ket{#1}\!\!\bra{#1}}

\newcommand{\phy}{\text{phy}}
\newcommand{\aux}{\text{aux}}

\newcommand{\abs}[1]{|#1|}

\begin{document}

\begin{CJK*}{UTF8}{mj}

\title{Quasi-Lindblad pseudomode theory for open quantum systems}

\author{Gunhee Park (박건희)}
\affiliation{Division of Engineering and Applied Science, California Institute of Technology, Pasadena, California 91125, USA}
\author{Zhen Huang}
\affiliation{Department of Mathematics, University of California, Berkeley, California 94720, USA}
\author{Yuanran Zhu}
\affiliation{Applied Mathematics and Computational Research Division, Lawrence Berkeley National Laboratory, Berkeley, California 94720, USA}
\author{Chao Yang}
\affiliation{Applied Mathematics and Computational Research Division, Lawrence Berkeley National Laboratory, Berkeley, California 94720, USA}
\author{Garnet Kin-Lic Chan}
\affiliation{Division of Chemistry and Chemical Engineering, California Institute of Technology, Pasadena, California 91125, USA}
\author{Lin Lin}
\affiliation{Department of Mathematics, University of California, Berkeley, California 94720, USA}
\affiliation{Applied Mathematics and Computational Research Division, Lawrence Berkeley National Laboratory, Berkeley, California 94720, USA}

\begin{abstract}
    We introduce a new framework to study the dynamics of open quantum systems with linearly coupled Gaussian baths. Our approach replaces the continuous bath with an auxiliary discrete set of pseudomodes with dissipative dynamics, but we further relax the complete positivity requirement in the Lindblad master equation and formulate a quasi-Lindblad pseudomode theory.
    We show that this quasi-Lindblad pseudomode formulation directly leads to a representation of the bath correlation function in terms of a complex weighted sum of complex exponentials, an expansion that is known to be rapidly convergent in practice and thus leads to a compact set of pseudomodes. The pseudomode representation is not unique and can differ by a gauge choice. When the global dynamics can be simulated exactly, the system dynamics is unique and independent of the specific pseudomode representation. However, the gauge choice may affect the stability of the global dynamics, and we provide an analysis of why and when the global dynamics can retain stability despite losing positivity.
    We showcase the performance of this formulation across various spectral densities in both bosonic and fermionic problems, finding significant improvements over conventional pseudomode formulations. 
\end{abstract}

\maketitle
\end{CJK*}

\section{Introduction}

Open quantum systems (OQS) are central to many fields, including quantum optics, condensed matter physics, chemical physics, and quantum information science~\cite{rivas2012open, breuer2007open, berkelbach2020open, LeggettRevModPhys.59.1}. The most widely studied case corresponds to a system coupled to a bath with continuous (bosonic or fermionic) degrees of freedom with the bath dynamics governed by a Hamiltonian quadratic in the bath field operators (a Gaussian bath) and with a system-bath coupling linear in the bath field operators. The primary task in such setups is to compute the system's dynamical observables.

There are two main classes of approaches to simulate this problem. The first approach integrates out the Gaussian bath dynamics to yield a path integral for the reduced system dynamics~\cite{FEYNMAN2000547}, which can then be evaluated using various approximations~\cite{Makri1995a, Strathearn2018, Ye2021, Ng2023, Muehlbacher2008, PhysRevLett.115.266802, NunezFernandez2022, ivander2024unifiedframeworkopenquantum}. The second approach, which is the focus of this work, replaces the continuous physical bath with a discrete mode representation~\cite{Vega2015, Prior2010, WoodsCramerPlenio2015, NusspickelBooth2020, Garraway1997, Dalton2001, Tamascelli2018, Schwarz2016, Lotem2020, Brenes2020, Zwolak2017, Trivedi2021}, i.e., an auxiliary bath, allowing the global state of the system and auxiliary bath to be computationally propagated in time. In both approaches, the influence of the bath on the reduced system dynamics is entirely determined by the bath correlation function (BCF)~\cite{FEYNMAN2000547, breuer2007open}. As long as the auxiliary BCF matches the original BCF, the reduced system dynamics can be shown to be identical~\cite{Tamascelli2018}. Therefore, the primary consideration in this second approach is to begin with a compact auxiliary bath representation of the BCF. 

The conventional discrete mode representation is to directly discretize the continuous Hamiltonian~\cite{Vega2015, Prior2010, WoodsCramerPlenio2015} without introducing any dissipation. However, this discretization requires a large number of modes in the long-time limit, which scales linearly with the propagated time~\cite{Vega2015}. 
In most physical cases, the continuous bath provides dissipation to the system.  It is thus natural to introduce an ansatz where the auxiliary bath dynamics is dissipative, in which case, the bath modes are referred to as pseudomodes~\cite{Garraway1997, Dalton2001, Tamascelli2018, Mascherpa2020, Chen_2019, LI2021127036}. The Liouvillian for the system and auxiliary bath can then be chosen to satisfy the constraint of physical dynamics, taking a Lindblad form and generating a completely positive trace-preserving (CPTP) dynamics. The corresponding BCF (assuming independent pseudomodes) is then a sum of \textit{complex} exponential functions weighted by \textit{positive} numbers. This BCF exactly represents the continuous spectral density as a positive weighted sum of Lorentzian functions. 
Unfortunately, many studies have found that this natural Lorentzian pseudomode expansion is inefficient, achieving only a slow rate of convergence as the number of pseudomodes increases~\cite{Mascherpa2020, Trivedi2021, Lednev2024}. Improving this convergence rate remains an active area of research, with approaches including the introduction of additional intra-bath Hamiltonian terms~\cite{Mascherpa2020, Lednev2024, Dorda2014, Dorda2015} or non-Hermitian pseudomode formulations~\cite{Lambert2019, Pleasance2020, Cirio2023, menczel2024nonhermitian}.

In the current work, we demonstrate that by relaxing the complete positivity (CP) of the global dynamics of the system and pseudomodes, we gain the flexibility to use a broader set of functions to approximate the BCF, significantly enhancing the convergence rate of the pseudomode expansion. The generator of the dynamics is a slight modification of the Lindblad form, which we refer to as a quasi-Lindblad pseudomode representation. The BCF is expressed as a sum of \textit{complex} exponentials weighted by \textit{complex} numbers. This ansatz is well-known in the signal processing and applied mathematics communities for yielding an exponentially fast convergence rate for many functions~\cite{beylkinMonzon2005, PaulrajRoyKailath1986, Fannjiang2020, Nakatsukasa2018aaa}, a result we also observe in our numerical tests. Additionally, this ansatz for the BCF has been adopted in non-pseudomode approaches to OQS dynamics, such as hierarchical equations of motion (HEOM)~\cite{Stockburger2022, xu2023universal, DanShi2023, ChenWangZhengetal2022, Takahashi2024} and tensor network influence functionals~\cite{vilkoviskiy2023bound, guo2024efficient}, providing a connection among these various methods.

It is worth noting that the BCF alone does not uniquely determine the pseudomode formulation. Meanwhile, the system dynamics is uniquely determined by the BCF and is independent of the pseudomode formulation as long as the global dynamics can be simulated exactly (i.e., without further numerical approximations such as truncating the number of bosonic modes or even finite precision arithmetic operations). However, 
certain pseudomode formulations can be unstable if the corresponding Liouvillian has eigenvalues with positive real parts, leading to exponentially growing modes. This phenomenon has also been observed in other contexts, such as HEOM~\cite{Krug2023, Dunn2019instabilities, Yan2020heomstability, LiYanShi2022}.
Therefore, in practice, the stability of the pseudomode formulation can indeed impact the quality of system dynamics. 
Nevertheless, with a careful choice of the pseudomode formulation that minimally violates the CP condition, we observe that the global dynamics can be stable across a range of physically relevant OQS in our simulations. We attribute this stability to the mixing of the divergent and dissipative parts of the global dynamics via the coherent terms in the Liouvillian.

\begin{figure}[t]
    \centering
    \includegraphics[width=1.0\columnwidth]{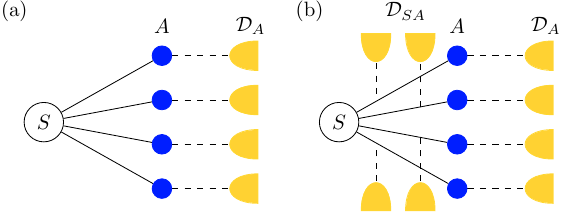}
    \caption{Pseudomode theory defines a discrete auxiliary dissipative bath $A$ to reproduce a system-reduced density operator. (a) The conventional pseudomode assumes bath-only dissipators, $\mathcal{D}_A$. (b) Our quasi-Lindblad pseudomode relaxes the complete positivity condition by introducing an additional system-bath Lindblad coupling with dissipators, $\mathcal{D}_{SA}$. 
     }
    \label{fig:fig_pseudomode}
\end{figure}

\section{Quasi-Lindblad pseudomode construction}

\subsection{Problem setup: original unitary OQS} 
\label{sec:unitary}

We consider a quantum system $S$ linearly coupled to a continuous Gaussian bath $B$, either bosonic or fermionic, evolving under the following Hamiltonian,
\begin{align}
    \hat{H}_{\phy} &= \hat{H}_S + \hat{H}_B + \hat{H}_{SB} \nonumber \\
     &= \hat{H}_S + \int \! d\omega \ \omega \ \hat{c}_{\omega}^\dag \hat{c}_{\omega} + \sum_j \hat{S}_{j} \hat{B}_j, \label{eq:ham}
\end{align}
where $\hat{H}_S$ denotes the system-only Hamiltonian and $\hbar=1$. 
$\hat{c}_{\omega}^\dag$ denotes the creation operator for the bath mode with energy $\omega$. In the system-bath coupling term, $\hat{H}_{SB}$, $\hat{S}_j$ and $\hat{B}_j$ are operators within the system and bath, respectively, and can be assumed to be Hermitian operators, without loss of generality~\cite{rivas2012open}. The bath coupling operator is given by $\hat{B}_j = \int d\omega \ g_{j}(\omega) \hat{c}_{\omega} + g_{j}^*(\omega) \hat{c}_{\omega}^\dag$. We assume the initial state to take a factorized form, $\hat{\rho}_{SB}(0) = \hat{\rho}_{S}(0) \otimes \hat{\rho}_{B} (0)$, where $\hat{\rho}_B(0)$ is a Gaussian state. We will also assume that $\hat{\rho}_B(0)$ is a number-conserving operator (and also for all initial bath reduced density operators henceforth) for simplicity.

Time-evolution of the density operator is described by the von Neumann equation, $\partial_t \hat{\rho}_{SB} = -i[\hat{H}_{\phy},\hat{\rho}_{SB}]$, and the system-reduced density operator is obtained by tracing out the bath, $ \hat{\rho}_S(t)=\tr_B[\hat{\rho}_{SB}(t)]$. The influence of the Gaussian bath on the system-reduced density operator is determined by the bath correlation function (BCF),
\begin{equation}\label{eq:bcf_def}
     C_{jj'}^{B}(t,t') = \tr_B \left[ \hat{B}_j(t) \hat{B}_{j'}(t') \hat{\rho}_B (0) \right],
\end{equation}
for $t \geq t' \geq 0$ with $\hat{B}_j(t) = e^{i\hat{H}_B t}\hat{B}_j e^{-i\hat{H}_B t}$. When $\hat{\rho}_B(0)$ is stationary, i.e., $e^{-i\hat{H}_B t}\hat{\rho}_B(0)e^{i\hat{H}_B t}=\hat{\rho}_B(0)$, the BCF only depends on the time difference $t-t'$, i.e., $C_{jj}^B(t, t') = C_{jj}^B(t-t')$.

\subsection{Conventional pseudomode} \label{sec:conventional_pseudomode}

In many cases of interest, the above unitary dynamics using a continuous bath leads to a dissipative bath correlation function, i.e., a decaying $C_{jj}^B(t- t')$ with $\lim_{t-t'\to \infty} C_{jj}^B(t-t') \to 0$.
To reproduce the associated dynamics of $\hat{\rho}_S(t)$ with a discrete bath, we now introduce a pseudomode description.
First, we review the conventional pseudomode formulation~\cite{Tamascelli2018}, which introduces an auxiliary bosonic bath $A$ with a Lindblad dissipator $\mathcal{D}_A$  applied within the bath $A$. The density operator across the system and auxiliary bath, $\hat{\rho}_{SA}$, evolves under the following Lindblad master equation,
\begin{equation}
    \partial_t \hat{\rho}_{SA} = -i \left[ \hat{H}_{\aux}, \hat{\rho}_{SA} \right]
    + \mathcal{D}_A   \hat{\rho}_{SA}, \label{eq:pseudomode_eom}
\end{equation}
where the Hamiltonian $\hat{H}_{\aux}=\hat{H}_S + \hat{H}_A + \hat{H}_{SA}$ is composed of the system (bath)-only Hamiltonian $\hat{H}_S$ ($\hat{H}_A$) and coupling Hamiltonian $\hat{H}_{SA}= \sum_j \hat{S}_{j} \hat{A}_j$. $\mathcal{D}_A$ has a Lindblad form 
\begin{equation}
    \mathcal{D}_A  \bullet = \sum_\mu \hat{L}_\mu \bullet \hat{L}_\mu^\dag - \frac{1}{2} \{\hat{L}_\mu^\dag  \hat{L}_\mu, \bullet \},
\end{equation}
which guarantees the CPTP property of the global dynamics consisting of the system and the pseudomodes. Here, $\hat{L}_\mu$ can be an arbitrary operator within the bath $A$. The 
Lindblad dissipator can be expressed in terms of a pseudomode basis,
\begin{equation}
    \mathcal{D}_A \bullet = 2 \sum_{kk'} \Gamma_{kk'} \left( \hat{F}_{k'} \bullet \hat{F}_{k}^\dag - \frac{1}{2} \{\hat{F}_{k}^\dag \hat{F}_{k'}, \bullet \} \right), 
\end{equation}
where $\hat{F}_k$ is an operator within the $k$th pseudomode, and $\Gamma_{kk'}$ is a positive semidefinite (PSD) matrix from the CPTP condition. $\hat{H}_A$ and $\mathcal{D}_A$ are quadratic in the creation and annihilation operators of the pseudomode, and $\hat{A_j}$ and $\hat{F}_k$ are linear, to maintain the Gaussian properties of the bath. $\mathcal{D}_A$ has negative eigenvalues, and thus, we can interpret it as adding dissipation to the reduced system dynamics.

The BCF for the pseudomode can be defined using the bath-only Liouvillian $\mathcal{L}_A \bullet = -i[\hat{H}_A, \bullet ] + \mathcal{D}_A \bullet$ and initial pseudomode reduced density operator $\hat{\rho}_A(0)$,
\begin{equation}
    C_{jj'}^{A}(t,t') =\tr_A\left[\hat{A}_j e^{\mathcal{L}_A (t-t')} \hat{A}_{j'} e^{\mathcal{L}_A t'} \hat{\rho}_A(0)\right].
    \label{eq:bcf_A}
\end{equation}
Ref.~\cite{Tamascelli2018} established the equivalence condition between the unitary bath dynamics and the pseudomode dynamics in terms of the BCF: if $C^B_{jj'}(t,t') = C^A_{jj'}(t,t')$ for $\forall t \geq t' \geq 0$, then $\tr_B[\hat{\rho}_{SB}(t)] = \tr_A [\hat{\rho}_{SA}(t)]$~\footnote{In the original formulation in \cite{Tamascelli2018}, the additional conditions on the one-point bath expectation values $\tr[\hat{B}_j(t)\hat{\rho}_B(0)]$ were imposed. In the main text, we further assumed that the initial bath reduced density operators, $\hat{\rho}_B(0)$ and $\hat{\rho}_A(0)$, are number-conserving operators. It leads to the one-point bath expectation values becoming zero for both the original unitary and pseudomode formulations.}. As in the unitary case in Sec.~\ref{sec:unitary}, when $\hat{\rho}_A(0)$ is stationary, i.e., $e^{\mathcal{L}_A t'} \hat{\rho}_A(0)=\hat{\rho}_A(0)$, the BCF depends solely on $t-t'$.

\subsection{Quasi-Lindblad pseudomode}

As described in the introduction, the conventional pseudomode description can reproduce a dissipative bath correlation function within a finite discretization, but it typically requires a large number of pseudomodes for an accurate approximation to the BCF. 
To remedy this, we now provide a more general pseudomode description that introduces an additional system-bath coupling $\mathcal{D}_{SA}$, expressed in terms of Lindblad dissipators,
\begin{equation}
    \mathcal{D}_{SA} \bullet = \sum_{j} \hat{L}'_{j} \bullet \hat{S}_j + \hat{S}_j \bullet \hat{L}'^\dag_{j} - \frac{1}{2}\{ \hat{S}_j \hat{L}'_{j} + \hat{L}'^\dag_{j}  \hat{S}_j, \bullet  \},
\end{equation}
where $\hat{L}'_j = \sum_k 2 M_{jk} \hat{F}_k$. We call this coupling a Lindblad coupling, as illustrated in Fig.~\ref{fig:fig_pseudomode}b. The total dissipator $\mathcal{D} = \mathcal{D}_A + \mathcal{D}_{SA}$ can be compactly written as
\begin{equation}
    \mathcal{D} \bullet = 2 \sum_{pq} \widetilde{\Gamma}_{pq} \left(\hat{F}_q \bullet \hat{F}^\dag_p - \frac{1}{2} \{\hat{F}^\dag_p \hat{F}_q, \bullet \} \right),
\end{equation}
where the index $p$ and $q$ refer to both the coupling index $j$ and the pseudomode bath index $k$, i.e., $p \in \{ j \} \cup \{ k \}$, and if $p=j$, $\hat{F}_j = \hat{S}_j$. The coefficient $\widetilde{\Gamma}_{pq}$ is
\begin{equation}
    \widetilde{\bm{\Gamma}} = \left(
    \begin{array}{cc}
    \bm{0} & \bm{M} \\
   \bm{M}^\dag & \bm{\Gamma}
    \end{array}
    \right),
\end{equation}
where the bold font denotes matrices. We note that the matrix $\widetilde{\Gamma}_{pq}$ is not PSD in general, unlike the PSD matrix $\Gamma_{pq}$, and hence, this pseudomode equation of motion violates the CP condition in the Lindblad master equation. Therefore, we call it a quasi-Lindblad pseudomode representation.

To take into account both the Hamiltonian and Lindblad coupling in the BCF, we define the system-bath Liouvillian $\mathcal{L}_{SA} \bullet = -i [\hat{H}_{SA}, \bullet] + \mathcal{D}_{SA} \bullet $ and further decompose it in the following form,
\begin{equation}
    \mathcal{L}_{SA} = -i \sum_j \mathcal{S}_j \mathcal{F}_j +i \sum_j \widetilde{\mathcal{S}}_j \widetilde{\mathcal{F}}_j ,
\end{equation}
where the system superoperators, $\mathcal{S}_j$ and $\widetilde{\mathcal{S}}_j$ are defined as $\mathcal{S}_j \bullet = \hat{S}_j \bullet$ and $\widetilde{\mathcal{S}}_j \bullet = \bullet \hat{S}_j $. The superoperators, $\mathcal{F}_j$ and $\widetilde{\mathcal{F}}_j$, are obtained as,
\begin{equation}
    \begin{gathered}
        \mathcal{F}_j \bullet = \hat{A}_j \bullet +  \bullet i \hat{L}'^\dag_{j} - \frac{i}{2}  (\hat{L}'_{j} + \hat{L}'^\dag_{j} ) \bullet, \\
        \widetilde{\mathcal{F}}_j \bullet = \bullet \hat{A}_j - i \hat{L}'_{j} \bullet +  \bullet \frac{i}{2} (\hat{L}'_{j} + \hat{L}'^\dag_{j} ).
    \end{gathered}
    \label{eq:F_superoperator}
\end{equation}
The BCF is defined with the superoperators $\mathcal{F}_j$, 
\begin{equation}
    C_{jj'}^{A}(t,t') =\tr_A\left[\mathcal{F}_j e^{\mathcal{L}_A (t-t')} \mathcal{F}_{j'} e^{\mathcal{L}_A t'} \hat{\rho}_A(0)\right].
    \label{eq:bcf_F}
\end{equation}
We remark that $\mathcal{F}_j$ reduces to $\hat{A}_j$ when $\hat{L}'_j=0$, and then the BCF in Eq.~\ref{eq:bcf_F} also reduces to the BCF in Eq.~\ref{eq:bcf_A}. Similar to the pseudomode description in Sec.~\ref{sec:conventional_pseudomode}, the above BCF leads to the equivalence of the reduced system dynamics to that generated by the original unitary formulation under the equivalence of the BCFs.
We provide the derivation of the equivalence condition in the Supplemental Material (SM)~\cite{supp_mat}.

In this framework, violation of the CP condition leads the dissipator Liouvillian $\mathcal{D}$ to have eigenvalues with positive real parts. Thus, in the absence of any coherent term (i.e., the commutator term of the form $-i[\hat{H}_{\rm aux},\bullet]$ for some Hamiltonian $\hat{H}_{\rm aux}$), this would lead to unstable global dynamics with modes growing exponentially in time. Note that when simulations are performed exactly, the reduced system dynamics, after tracing out the bath, can still be perfectly reproduced from the BCF equivalence condition. However, this cannot be guaranteed in practical simulations involving finite precision arithmetic operations and physical approximations, such as truncating the number of bosonic modes. Therefore, the stability of the Liouvillian is an important issue to study.

Interestingly, we find that the Liouvillian, including both the coherent and dissipator terms, \textit{can} generate stable dynamics despite violating the CP condition. We will analyze this effect in more detail in Sec.~\ref{sec:sec4}.

\subsection{Ansatz for quasi-Lindblad pseudomode}

We next introduce an ansatz for the quasi-Lindblad pseudomode that yields a simple form for the BCF, facilitating numerical fitting. We describe this ansatz for both bosonic and fermionic baths. We mainly target the reproduction of the BCF of the unitary dynamics in Sec.~\ref{sec:unitary} with the initial bath being a stationary state, 
such as a thermal state $\hat{\rho}_B(0) \propto e^{-\beta \hat{H}_B}$, where $\beta$ represents the inverse temperature.

We start with a bosonic bath, setting the initial reduced density operator to a vacuum state, $\hat{\rho}_A(0) = \proj{\bm{0}}$. The dissipator $\mathcal{D}_A$ is set as $\hat{F}_k = \hat{d}_k$, where $\hat{d}_k$ is the bosonic annihilation operator of the $k$th pseudomode, and the Hamiltonian $\hat{H}_A = \sum_{kk'} H^A_{kk'} \hat{d}_k^\dag \hat{d}_{k'}$. This ansatz satisfies the stationary condition $\mathcal{L}_A \hat{\rho}_A(0)=0$. Then, $e^{\mathcal{L}_A t'} \hat{\rho}_A(0) = \hat{\rho}_A(0)$ and $C^A_{jj'}(t,t')$ is only dependent on the time difference $C^A_{jj'}(t,t') = C^A_{jj'}(t-t') = C^A_{jj'}(\Delta t)$, where $\Delta t = t - t'$. The system-bath coupling is set as $\hat{A}_j \ = \sum_k V_{jk} \hat{d}_k + V_{jk}^* \hat{d}_k^\dag$.

We obtain the BCF based on the above ansatz:
\begin{equation}
    \bm{C}^A(\Delta t) = (\bm{V}-i\bm{M}) e^{-\bm{Z}\Delta t} (\bm{V}+i\bm{M})^\dag,
    \label{eq:Corrt_multisite}
\end{equation}
where $\bm{Z} = \bm{\Gamma} +i\bm{H}^A $. This BCF expression has a gauge redundancy: $\bm{Z} \rightarrow \bm{G}\bm{Z}\bm{G}^{-1}$, $\bm{V}-i\bm{M} \rightarrow (\bm{V}-i\bm{M}) \bm{G}^{-1}$, $\bm{V} + i \bm{M} \rightarrow (\bm{V}+i\bm{M}) \bm{G}^{\dagger}$, for arbitrary invertible matrix $\bm{G}$. Therefore, we focus on a diagonal matrix $\bm{Z}=\mathrm{diag}(\{z_k\})$, which can be transformed to a more general diagonalizable $\bm{Z}$ through a unitary transformation~\footnote{In principle, we have a general expression by considering non-diagonalizable $\bm{Z}$. However, numerically, fitting with arbitrary $\bm{Z}$ yields the same optimal solution as with diagonalizable $\bm{Z}$.}. In this case, we define $\hat{H}_A = \sum_k E_k \hat{d}^\dag_k \hat{d}_k$ and $\Gamma_{kk'}=\gamma_k \delta_{kk'}$, and  $z_k = \gamma_k+iE_k$. Even with a diagonal matrix $\bm{Z}$, $\bm{V}$ and $\bm{M}$ still have a gauge redundancy associated with a diagonal $\bm{G}$, since $\bm{Z}$ remains invariant under the transformation $\bm{G}\bm{Z}\bm{G}^{-1} = \bm{Z}$.
It is important to note that while different choices of gauge yield the same BCF and thus the same reduced system dynamics (in the absence of simulation errors), the total system-bath Liouvillian differs, leading to different global system-bath dynamics.

The case of a fermionic bath can be introduced similarly.  In particular, we focus on number-conserving fermionic impurity models where the coupling Hamiltonian is given by $\hat{H}_{SB} = \sum_j \hat{a}_j^\dag \hat{B}_j + \hat{B}_j^\dag \hat{a}_j, \: \hat{B}_j = \int d\omega \ g_{j}(\omega) \hat{c}_\omega$ ($\hat{a}_j^\dag$ is the creation operator of the impurity with index $j$, which can involve both spin and orbital degrees of freedom). The number-conserving property of the fermionic impurity model reduces the BCFs to two different BCFs, namely the greater and lesser BCFs: $C^{>B}_{jj'}$ and $C^{<B}_{jj'}$. For any fermionic Gaussian state $\hat{\rho}_B(0)$, $C^{>B}_{jj'}$ and $C^{<B}_{jj'}$ are defined as,
\begin{equation}
    \begin{gathered}
        C^{>B}_{jj'}(t,t') = \tr_B\left[\hat{B}_j(t) \hat{B}^{\dag}_{j'}(t') \hat{\rho}_B(0) \right], \\
        C^{<B}_{jj'}(t,t') = \tr_B\left[\hat{B}^\dag_j(t) \hat{B}_{j'}(t') \hat{\rho}_B(0) \right].
    \end{gathered}
\end{equation}
Corresponding to these two BCFs, we use two different auxiliary baths with an initial reduced density operator $\hat{\rho}_A(0)=\hat{\rho}_{A_1}(0)\otimes \hat{\rho}_{A_2}(0) = \proj{\bm{0}}\otimes\proj{\bm{1}}$. The bath-only Hamiltonian and dissipator are set to be diagonal without loss of generality, as discussed in the bosonic case. The Hamiltonian is $\hat{H}_A = \sum_{k_1}E_{k_1} \hat{d}^\dag_{k_1} \hat{d}_{k_1} + \sum_{k_2} E_{k_2} \hat{d}^\dag_{k_2} \hat{d}_{k_2}$ where $k_1$ and $k_2$ refer to the pseudomode index in each bath, respectively, and the dissipator is set to be 
\begin{equation}
    \begin{gathered}
        \mathcal{D}_{A_1} \bullet = \sum_{k_1} 2\gamma_{k_1} \left(\hat{d}_{k_1} \bullet \hat{d}^{\dag}_{k_1} - \frac{1}{2} \{\hat{d}^{ \dag}_{k_1}  \hat{d}_{k_1}, \bullet \}\right), \\
        \mathcal{D}_{A_2} \bullet = \sum_{k_2} 2\gamma_{k_2} \left(\hat{d}^\dag_{k_2} \bullet \hat{d}_{k_2} - \frac{1}{2} \{\hat{d}_{k_2}  \hat{d}^\dag_{k_2}, \bullet \}\right),
    \end{gathered}
\end{equation}
which satisfies $\mathcal{D}_{A_i} \hat{\rho}_{A_i}(0) = 0$ for both $i=1,2$. The system-bath Hamiltonian coupling is $\hat{H}_{SA} = \hat{a}_j^\dag \hat{A}_j + \hat{A}_j^\dag \hat{a}_j$, $\hat{A}_j = \sum_{k_1} (\bm{V}_1)_{jk_1} \hat{d}_{k_1} + \sum_{k_2} (\bm{V}_2)_{jk_2} \hat{d}_{k_2}$ and the system-bath Lindblad coupling is set to be
\begin{equation}
    \begin{gathered}
        \mathcal{D}_{SA_1} \bullet = \sum_j \hat{a}_j \bullet \hat{L}_{1j}^\dag + \hat{L}_{1j} \bullet \hat{a}_j^\dag - \frac{1}{2} \{\hat{L}_{1j}^\dag \hat{a}_j+\hat{a}_j^\dag \hat{L}_{1j},\bullet \}, \\
        \mathcal{D}_{SA_2} \bullet = \sum_j \hat{a}^\dag_j \bullet \hat{L}_{2j} + \hat{L}^\dag_{2j} \bullet \hat{a}_j - \frac{1}{2} \{\hat{L}_{2j} \hat{a}^\dag_j+\hat{a}_j \hat{L}^\dag_{2j},\bullet \},
    \end{gathered}
\end{equation}
where $\hat{L}_{1j} = \sum_{k_1} 2(\bm{M}_1)_{jk_1} \hat{d}_{k_1}$ and $\hat{L}_{2j} = \sum_{k_2} 2(\bm{M}_2)_{jk_2} \hat{d}_{k_2} $. The BCF with this ansatz becomes
\begin{equation}
    \begin{gathered}
        \bm{C}^{>A}(\Delta t) = (\bm{V}_{1}-i\bm{M}_{1}) e^{-\bm{Z}_{1}\Delta t} (\bm{V}_{1}+i\bm{M}_{1})^\dag, \\
        \bm{C}^{<A}(\Delta t) = (\bm{V}_{2}-i\bm{M}_{2})^* e^{-\bm{Z}_{2}\Delta t} (\bm{V}_{2}+i\bm{M}_{2})^T, 
    \end{gathered}
\end{equation}
where $\bm{Z}_{1}$ and $\bm{Z}_{2}$ are diagonal matrices with elements $z_{k_1} = \gamma_{k_1} + iE_{k_1}$ and $z_{k_2} = \gamma_{k_2}-iE_{k_2}$. We note that the baths $A_1$ and $A_2$ only contribute to the BCF $\bm{C}^{>A}$ and $\bm{C}^{<A}$, respectively. 
We remark that the BCF from the fermionic bath is also expressed as a complex weighted sum of complex exponential functions and has the same gauge redundancy as the bosonic bath.

It is instructive to consider a model with a single coupling term (with only $j=1$). In this case, the BCF simplifies to a scalar function without indices:
\begin{equation}
    C^A(\Delta t) = \sum_k w_k e^{-z_k \Delta t},
\end{equation}
where $w_k = (V_k - iM_k)(V_k^* - i M_k^*)$. This expression is a sum of \emph{complex} exponential functions with \emph{complex} weights. In contrast, the conventional pseudomode description corresponds to $M_k = 0$, resulting in \emph{positive} weights $w_k = |V_k|^2$.
Physically, this exactly represents the BCF from a positive weighted sum of the Lorentzian spectrum, and thus, we also refer to the conventional pseudomode as a Lorentzian pseudomode. We see, however, that the quasi-Lindblad pseudomode provides for a more flexible and general representation, whose influence on the compactness of the pseudomode description we will examine in the numerics below.

In addition to the enhanced flexibility, the use of complex exponential expressions also simplifies the numerical fitting of the BCF. In this work, we utilize the ESPRIT algorithm~\cite{PaulrajRoyKailath1986, Fannjiang2020}, which has previously been reported as one of the most competitive methods for BCF fitting~\cite{Takahashi2024}, to determine the complex exponents, $z_k$.
We then perform a least-squares fitting of the complex weights $w_k$.
We note that this simple unconstrained least-squares approach is possible because we do not restrict $w_k \geq 0$ as would be required in the conventional pseudomode approach. Once $w_k$ is determined, we choose the gauge for $V_k$ and $M_k$. In this limit, the gauge choice is parameterized as $V_k - i M_k = \kappa_k \sqrt{w_k}, V_k + i M_k = \kappa_k^{*-1} \sqrt{w^*_k}$ where $\kappa_k \in \mathbb{C}$.
Here, we follow a simple heuristic of minimizing the magnitude of $M_k$, as this term determines the extent of CP condition violation. We describe the result for the bosonic bath expression in Eq.~\ref{eq:Corrt_multisite} for $\bm{V}$ and $\bm{M}$, and its adaptation to the fermionic BCF expression is straightforward. 
The solution of $V_k$ and $M_k$ with minimum $|M_k|$  is then given by $V_k - iM_k = \sqrt{w_k}$ with $V_k, M_k \in \mathbb{R}$, ($\kappa_k = 1$) which we provide its proof in the SM~\cite{supp_mat}.  We will examine other choices of $V_k$ and $M_k$ and their effect on the stability of the system-bath dynamics in Sec.~\ref{sec:sec4}. 

There are several approaches to relax the CPTP condition in pseudomode representations. One approach involves a non-Hermitian pseudomode formulation~\cite{Lambert2019, Pleasance2020, Cirio2023, menczel2024nonhermitian}, where the system-bath coupling is given by a non-Hermitian Hamiltonian (without Lindblad coupling). This formulation results in real weights, $w_k \in \mathbb{R}$~\cite{Lambert2019, menczel2024nonhermitian}, or complex weights from a pair of pseudomodes~\cite{menczel2024nonhermitian}. Another approach, presented in \cite{xu2023universal}, is a Lindblad-like pseudomode formulation derived from the HEOM with a similarity transformation. While this approach allows for complex weights,
it is important to emphasize that the Lindblad-like equation in \cite{xu2023universal} represents one gauge choice for $\kappa_k$ \footnote{The corresponding $\kappa_k$ for the Lindblad-like equation in \cite{xu2023universal} is $\kappa_k = \sqrt{\operatorname{Re}[w_k]/w_k}$.} within our quasi-Lindblad pseudomode framework.

\begin{figure*}[t]
    \centering
    \includegraphics[width=0.95\textwidth]{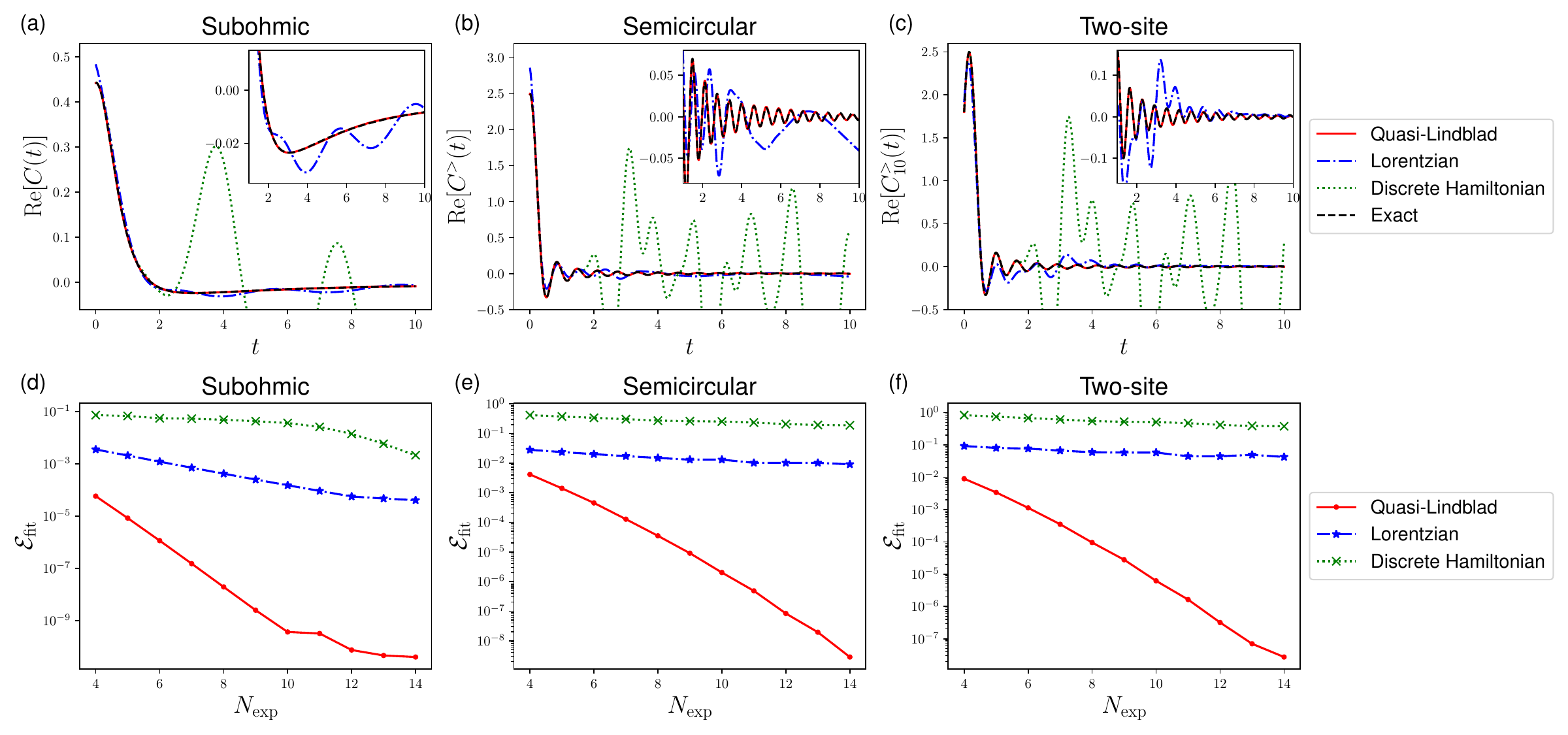}
    \caption{Exponential function fitting of BCF for (a,d) subohmic, (b,e) semicircular, (c,f) and two-site impurity spectral densities. The real part of the fitted BCF is shown in (a,b,c) from the complex-weighted fit to complex exponentials (red solid, quasi-Lindblad pseudomode formulation in this work), positive-weighted fit to complex exponentials (blue dash-dotted, conventional Lorentzian pseudomode formulation), and positive-weighted fit to real exponentials (green dotted, discrete Hamiltonian) compared to the exact BCF (black dashed). The inset shows the zoomed-in data at $t>1$ (only the results from complex exponential fits). The number of exponential functions, $N_{\exp}$, used for the fitting was $N_{\exp}=4$ and $6$ for (a) and (b,c), respectively. (d,e,f) Fitting errors, $\mathcal{E}_{\text{fit}}$, as a function of $N_{\exp}$ from the fitting with complex weights (red dots), positive weights (blue stars), and discrete Hamiltonian (green crosses). 
    }
    \label{fig:fitting}
\end{figure*}

For general couplings with multiple $j$, a similar strategy can be adopted. In this case, we first find a sum of exponential functions with a matrix-valued weight $\bm{C}^A(\Delta t) \approx \sum_k \bm{W}_k e^{-z_k \Delta t}$. 
We determine $z_k$ by fitting a scalar quantity, such as the sum of the BCF entries, $\sum_{jj'} C^A_{jj'}(\Delta t)\approx \sum_k \left( \sum_{jj'}  (\bm{W}_k)_{jj'}\right) e^{-z_k \Delta t}$, or the trace $\sum_{j} C^A_{jj}(\Delta t)\approx \sum_k \left( \sum_{j}  (\bm{W}_k)_{jj}\right) e^{-z_k \Delta t}$. The matrix-valued complex weights $\bm{W}_k$ are then obtained through the least-squares fitting. 

Finding $\bm{V}$ and $\bm{M}$ from $\bm{W}_k$ involves additional matrix decompositions. However, the representation from Eq.~\ref{eq:Corrt_multisite} has a rank-1 matrix factorization form, $(\bm{W}_k)_{jj'}=(\bm{V}-i\bm{M})_{jk} (\bm{V}+i\bm{M})^*_{j'k}$, and $\bm{W}_k$ from the least-squares fitting is generally not a rank-1 matrix, but of rank $N_S$, after denoting the range of $j$ as $N_S$. To address this, we represent the rank-$N_S$ $\bm{W}_k$ by indexing the pseudomode with $n=(k,s)$ where the index $s$ is an additional index with a range of $N_S$. Then, we assign the diagonal elements of $\bm{Z}$ as $z_n = z_{(k,s)}=z_k$. This gives the following expression: 
\begin{equation}
    C^A_{jj'}(\Delta t) = \sum_{n=(k,s)} (\bm{V} - i \bm{M})_{jn} (\bm{V} + i \bm{M})^*_{j'n} e^{-z_k \Delta t},
\end{equation}
where $(\bm{W}_k)_{jj'} = \sum_s (\bm{V} - i \bm{M})_{jn} (\bm{V} + i \bm{M})^*_{j'n}$. Thus, given $N_{\exp}$ exponents $z_k$ and $N_S$ coupling operators, $N_S N_{\exp}$ pseudomodes are used in the BCF fitting. With this setting, we decompose $\bm{W}_k$ with a singular value decomposition (SVD), $\bm{W}_k = \bm{U}_k \bm{\Sigma}_k \widetilde{\bm{U}}_k^\dag$, where $\bm{\Sigma}_k$ is a diagonal matrix. One possible choice for $\bm{V}$ and $\bm{M}$ is $(\bm{V}-i\bm{M})_{jn} = (\bm{U}_k \bm{\Sigma}^{1/2}_k)_{js}, (\bm{V} + i \bm{M})_{jn} = (\widetilde{\bm{U}}_k {\bm{\Sigma}_k}^{1/2})_{js}$, which reduces to the gauge choice above for the single coupling limit~\footnote{The SVD representation has a phase redundancy in $\bm{U}_k$, $\bm{U}_k \rightarrow \bm{U}_k e^{i \bm{\Theta}_k}$, and also in $\bm{M}_k$, $\bm{M}_k \rightarrow \bm{M}_k e^{i \bm{\Theta}_k}$, where $\bm{\Theta}_k$ is a real diagonal matrix. However, it does not affect the overall norm of $\bm{M}_k$.}.

\section{Numerical results}

\subsection{Fitting the bath correlation function}

The quasi-Lindblad pseudomode formulation provides the BCF as a complex weighted sum of complex exponential functions. 
This type of BCF representation is widely used, for example, in HEOM~\cite{Stockburger2022, DanShi2023, xu2023universal, ChenWangZhengetal2022, Takahashi2024} and tensor network influence functionals~\cite{vilkoviskiy2023bound, guo2024efficient}. It is known in many applications to show a fast exponential convergence rate with the number of pseudomodes. In this section, we focus on comparing the performance of BCF fitting using the quasi-Lindblad ansatz to that achieved with the conventional Lorentzian pseudomode ansatz and a discrete Hamiltonian ansatz obtained from direct discretization.

Given a spectral density $J_{jj'}(\omega) = g_{j}(\omega) g_{j'}^*(\omega)$, we will fit the associated zero-temperature BCF~\cite{breuer2007open}, 
\begin{equation}
    C_{jj'}(t) = \int J_{jj'}(\omega) e^{-i\omega t} d\omega,
\end{equation}
using the three different ansatz. We examine the performance of these ansatz for three spectral densities, as illustrated in Fig.~\ref{fig:fitting}. The first case is a subohmic spectral density, 
\begin{equation}\label{eq:subohmic}
    J(\omega) = \frac{\alpha}{2} \omega_c^{1-s} \omega^s e^{-\omega/\omega_c}, \quad \omega \in [0,\infty),
\end{equation}
the second, a semicircular spectral density, 
\begin{equation}\label{eq:semicircular}
    J(\omega) = \frac{\Gamma}{\pi}\sqrt{1 - \frac{\omega^2}{W^2}}, \quad \omega \in [0,W],
\end{equation}
and the third corresponds to a two-site spectral density,
\begin{equation}\label{eq:spectraldensity_twosite}
    J_{jj'}(\omega) = \begin{pmatrix}
        1 & r(\omega) \\ r^*(\omega) & 1
        \end{pmatrix} J(\omega),
\end{equation}
where $r(\omega) = \exp(-i \omega /2W)$ and $J(\omega)$ is the semicircular spectral density specified above.
For the subohmic case, we use $s = 0.5$ and $\alpha = \omega_c=1$, and for the semicircular case, we set $\Gamma=1$ and $W=10\Gamma$. 

Figs.~\ref{fig:fitting} (a--c) compare the fitted BCF to the exact expressions shown by the black dashed lines. Three different fits are compared: (1) complex exponential function fits with complex weights (red solid), (2) with positive (or PSD for the multi-site case) weights (blue dash-dotted), and (3) a discrete Hamiltonian (green dotted).
Fitting scheme (2) represents the Lorentzian pseudomode.
Here, the parameters are optimized through numerical minimization of the fitting error. For fitting (3), we used a discretization scheme based on Gaussian quadrature using Legendre polynomials~\cite{Vega2015}. Details of these fitting schemes are in the SM~\cite{supp_mat}. 
The fits to the quasi-Lindblad ansatz with $N_{\exp}=4$ ($N_{\exp}=6$) for Fig.~\ref{fig:fitting}a (Fig.~\ref{fig:fitting}b and Fig.~\ref{fig:fitting}c), respectively, show high accuracy already with a small number of exponentials. 
In Fig.~\ref{fig:fitting}d, Fig.~\ref{fig:fitting}e, and \ref{fig:fitting}f, the maximum fitting error as a function of $N_{\text{exp}}$ is illustrated. This comparison shows the clear advantage of the quasi-Lindblad pseudomode compared to the Lorentzian pseudomode and direct discretization schemes and illustrates an exponential decrease of error with respect to $N_{\text{exp}}$ in both cases. 

\begin{figure}[t]
    \centering
    \includegraphics[width=0.8\columnwidth]{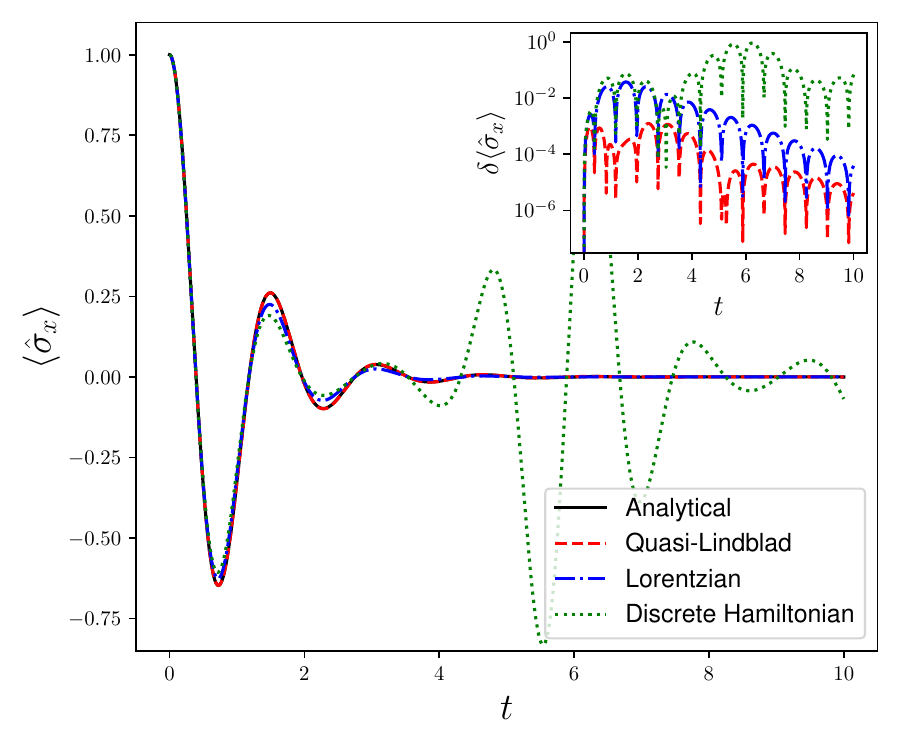}
    \caption{Time evolution of coherence $\langle \hat{\sigma}_x \rangle$ from the quench dynamics of the spin-boson model. The analytical result from the continuous bath with $s=0.5$ subohmic spectral density is compared to results from a two-mode discretized bath using quasi-Lindblad pseudomodes, Lorentzian pseudomodes, and a discrete Hamiltonian. \textbf{(inset)} Errors of coherence compared to the analytical result.}
    \label{fig:quench_spinboson}
\end{figure}

\subsection{Quench dynamics simulation}\label{sec:sec3b}

We then test the accuracy of the quasi-Lindblad pseudomode formulation in simulating the quench dynamics of a spin-boson model and a fermionic impurity model. First, we demonstrate the effectiveness of this approach in the spin-boson model. We choose the Hamiltonian $\hat{H}_S = \omega_0 \hat{\sigma}_z/2$ with $\omega_0=4$ and $\hat{S}_{j=1} = \hat{\sigma}_z$ and the initial system-reduced density operator $\hat{\rho}_S(0) = \proj{+}$, with $\ket{+}=1/\sqrt{2} (\ket{0}+\ket{1})$. While the analytical solution for this setup is known~\cite{breuer2007open, LeggettRevModPhys.59.1}, obtaining accurate numerical results remains nontrivial~\cite{Somoza2019, Mascherpa2020, Cygorek2024prx}, 
particularly due to the requirement for a precise description of the BCF.

We choose the $s=0.5$ subohmic spectral density as above ($\alpha=\omega_c=1$) and consider a zero temperature bath initial condition. The BCF is fitted using two bosonic modes, and we compare the quasi-Lindblad pseudomode, Lorentzian pseudomode, and a discrete Hamiltonian from Gaussian quadrature. Fig.~\ref{fig:quench_spinboson} shows the time evolution of the coherence $\langle \hat{\sigma}_x(t) \rangle$ using these three different bath discretizations. We see that the time evolution from the discrete Hamiltonian shows unphysical oscillations in the long-time limit, while the Lorentzian pseudomode qualitatively captures the dissipative dynamics but at a much lower accuracy than the quasi-Lindblad pseudomode representation.

\begin{figure}[t]
    \centering
    \includegraphics[width=1.0\columnwidth]{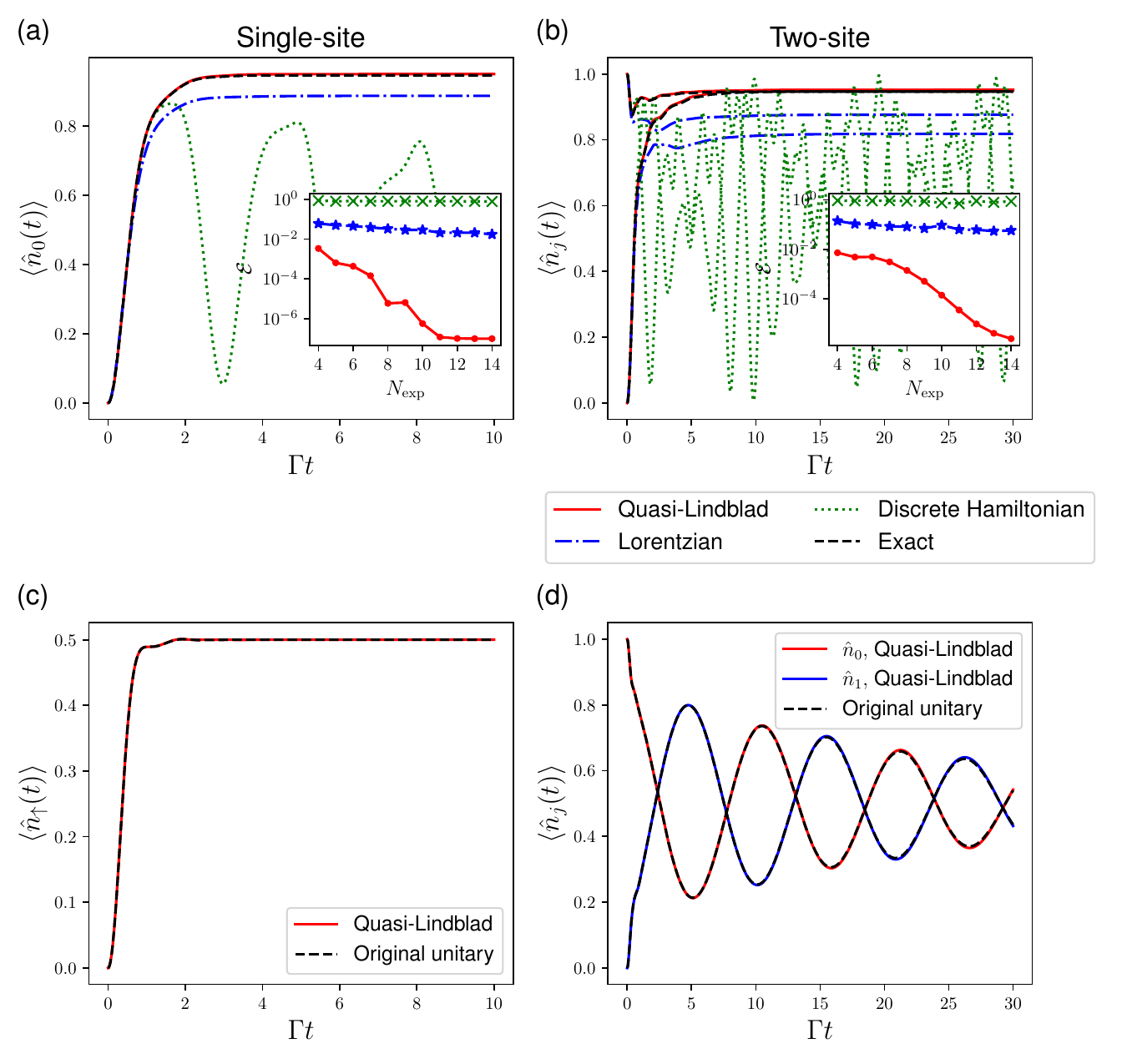}
    \caption{Quench dynamics of spinless (a) single-site and (b) two-site impurity models with $\hat{H}_S=\varepsilon \sum_j \hat{n}_j $ with $\varepsilon= -4\Gamma$, (c) spinful SIAM and (d) spinless two-site impurity model with $\hat{H}_S= U \hat{n}_0 \hat{n}_1 + \varepsilon \sum_j \hat{n}_j $ with $U=-2\varepsilon=8\Gamma$.
    (a, b) Time-dependence of impurity occupations, $\langle\hat{n}_j(t)\rangle$, computed from Gaussian calculations from our quasi-Lindblad pseudomode, Lorentzian pseudomode, discrete Hamiltonian, and exact unitary dynamics with $N_{\text{exp}}=4$. \textbf{(inset)} A maximum error of $\langle \hat{n}_j(t)\rangle$ from pseudomode dynamics as a function of $N_{\text{exp}}$ for three different fitting schemes in Fig.~\ref{fig:fitting}. (c, d) $\langle\hat{n}_j(t)\rangle$ of our pseudomode compared to the original unitary dynamics with large bath discretization, computed from tdDMRG. }
    \label{fig:quench}
\end{figure}

We next study the dynamics of the spinless noninteracting single-site and two-site models with $\hat{H}_S= \varepsilon \sum_j \hat{n}_j$ ($\hat{n}_j = \hat{a}_j^\dag \hat{a}_j$). Here, we can carry out efficient Gaussian calculations, which we do for both a continuous unitary bath and for the Lindblad pseudomode dynamics~\cite{Lotem2020, Barthel_2022}. We choose the semicircular spectral density in Eq.~\ref{eq:semicircular} and Eq.~\ref{eq:spectraldensity_twosite} with $W=10 \Gamma$, but with $\omega \in [-W,W]$, $\varepsilon = -4\Gamma$, and consider the initial bath as a fully occupied (empty) state for $\omega <0$ ($\omega >0$). We fit the two BCFs, $C^{>}_{jj'}(t)$ and $C^{<}_{jj'}(t)$, with $N_{\text{exp}}$ exponential functions each and represent the bath with in total $2N_{\text{exp}}$ ($4N_{\text{exp}}$) pseudomodes for the single-site (two-site) impurity model, respectively. 
The initial impurity state is assumed to be an empty state, $\hat{\rho}_S(0)= \proj{0}$, for the single-site model, and $\hat{\rho}_S(0) = \proj{10}$ for the two-site model. Fig.~\ref{fig:quench}a and \ref{fig:quench}b show the time-dependence of the impurity occupation, $\langle \hat{n}_j(t) \rangle$, with $N_{\text{exp}}=4$ for the exact dynamics, the quasi-Lindblad pseudomode, Lorentzian pseudomode, and discrete Hamiltonian dynamics. These results clearly show the improved accuracy of the quasi-Lindblad pseudomode relative to the Lorentzian pseudomode, which reaches the wrong steady state.
The inset in Fig.~\ref{fig:quench}a and \ref{fig:quench}b shows the maximum error of the impurity occupation as a function of $N_{\text{exp}}$, which also shows an exponential convergence in $N_{\text{exp}}$.

\begin{figure*}
    \centering
    \includegraphics[width=0.9\textwidth]{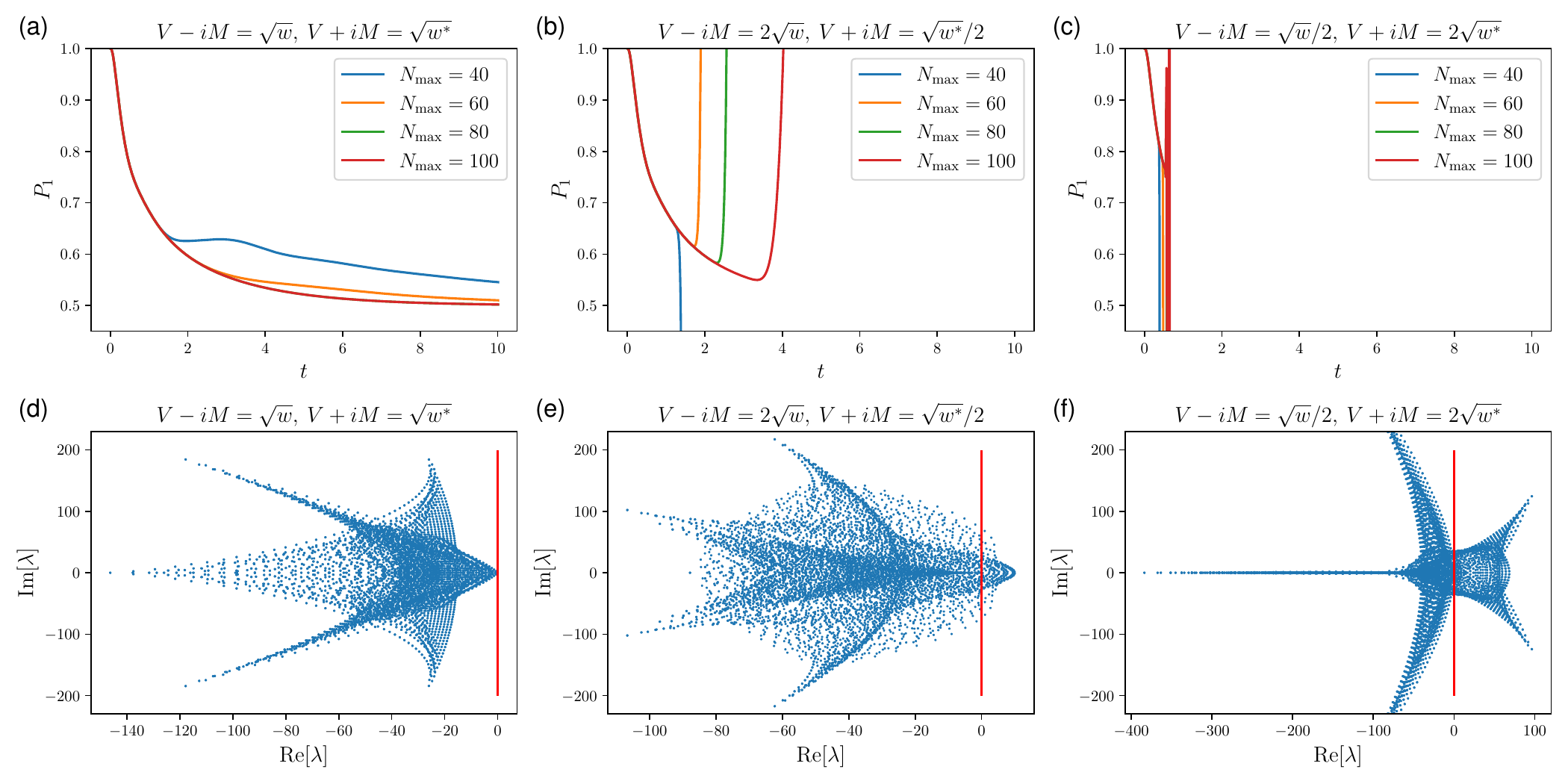}
    \caption{Stability behavior of the spin-boson model for different gauge choices (for the same BCF) in the quasi-Lindblad formulation (a,d) $\kappa=1$, (b,e) $\kappa=2$, and (c,f) $\kappa=0.5$, (details in the main text). (a,b,c) Time evolution of the excited state occupation, $P_1(t) =  \bra{1} \hat{\rho}_S(t) \ket{1}$. (d,e,f) Real and imaginary parts of eigenvalues $\lambda$ of Liouvillian $\mathcal{L} = \mathcal{H} + \mathcal{D}$ at $N_{\max}=40$. The red vertical line marks the point where the real part of the eigenvalue is zero. A positive real part in the eigenvalues implies instability. 
    } 
    \label{fig:instability}
\end{figure*}

Next, we carry out a test with interacting models - a spinful single-site impurity model, known as the single-impurity Anderson model (SIAM), with $\hat{H}_S = U \hat{n}_{\uparrow} \hat{n}_{\downarrow} + \varepsilon ( \hat{n}_{\uparrow} + \hat{n}_{\downarrow})$, and a spinless two-site impurity model with $\hat{H}_S = U \hat{n}_{0} \hat{n}_{1} + \varepsilon ( \hat{n}_{0} + \hat{n}_{1})$ - and choose $U = -2\varepsilon =8\Gamma$. 
For the SIAM model, we assume that there is the same bath for each spin with a semicircular spectral density without any coupling between the baths. The initial bath state is assumed to be the same as in the noninteracting case. As a reference for the exact unitary dynamics, we use a large bath discretization obtained from Gaussian quadrature with Legendre polynomials~\cite{Vega2015} - 80 (200) bath modes for the single-site (two-site) impurity model, respectively. For both the unitary reference and pseudomode Lindblad dynamics, we use the time-dependent density matrix renormalization group (tdDMRG) to propagate the dynamics~\cite{Feiguin2004, PAECKEL2019167998} with bond dimensions of up to 250, using the Block2~\cite{Zhai2021, Zhai2023} package. 
Fig.~\ref{fig:quench}c and \ref{fig:quench}d show the time-dependence of $\langle \hat{n}_{j}(t) \rangle$ ($j=\uparrow$ for the SIAM, $j=0,1$ for the two-site impurity). The quasi-Lindblad ansatz was constructed from $N_{\text{exp}}=4$ ($2N_{\text{exp}}=8$ pseudomodes per each spin) for SIAM and $N_{\text{exp}}=8$ ($4N_{\text{exp}}=32$ pseudomodes) for the two-site impurity model. We see that the quasi-Lindblad pseudomode accurately reproduces the original unitary dynamics, including the long-lived oscillation in the two-site impurity model.

\section{Stability of Quasi-Lindblad dynamics}\label{sec:sec4}

We now discuss the stability of quasi-Lindblad global dynamics. As mentioned above, the quasi-Lindblad master equation relaxes the CP condition on the Lindblad dissipator forms, and this implies that the dissipator for the quasi-Lindbladian, $\mathcal{D} = \mathcal{D}_A + \mathcal{D}_{SA}$, always has an eigenvalue with a positive real part. However, we observe that adding the Hamiltonian part of the Liouvillian, $\mathcal{H}\bullet=-i[\hat{H}_{\aux},\bullet]$, can induce a stable dynamics (i.e., the real parts of all eigenvalues of $\mathcal{L} = \mathcal{H}+\mathcal{D}$ are nonpositive). We call this Hamiltonian-induced stability. For some intuition, we can consider the Trotter decomposition of the time evolution operator $e^{\mathcal{L} t} = \lim_{N \rightarrow \infty} \left(e^{\mathcal{H}t/N} e^{\mathcal{D}t/N} \right)^N$. Although the time evolution $e^{\mathcal{D}t/N}$ contains both dissipative and divergent components, the Hamiltonian part $e^{\mathcal{H}t/N}$ can mix these different components so that the overall dynamics \textit{may} become stable. 

As a concrete first example, consider the following quasi-Lindblad dynamics for a single spin
\begin{equation}\label{eqn:quasi_bloch}
\partial_t \hat{\rho}=-i[\hat{H},\hat{\rho}]+\gamma_x\mathcal{L}_{\hat{\sigma}_x}(\rho)-\gamma_z\mathcal{L}_{\hat{\sigma}_z}(\hat{\rho}).
\end{equation}
Here $\mathcal{L}_{\hat{\sigma}_\alpha}\bullet=\hat{\sigma}_\alpha \bullet \hat{\sigma}_\alpha- \bullet $ where $\hat{\sigma}_\alpha$ is a Pauli matrix. Without the Hamiltonian part, the dynamics is clearly unstable for any $\gamma_x,\gamma_z>0$. However, if we choose $\hat{H}=\nu \hat{\sigma}_y$ and solve Eq.~\eqref{eqn:quasi_bloch} using the Bloch sphere representation $\hat{\rho}=\frac12(\hat{I}+{\bm a}\cdot \bm{\hat{ \sigma}})$, we can write down a set of closed equations for the coefficients ${\bm a}(t)$ as
\begin{equation}
\partial_t {\bm a}(t)=2\begin{pmatrix}
\gamma_z & 0 & \nu \\
0 & -\gamma_x+\gamma_z & 0\\
-\nu & 0 & -\gamma_x 
\end{pmatrix} 
{\bm a}(t).
\end{equation}
Therefore, when $\gamma_x>\gamma_z>0, \abs{\nu}^2>\gamma_x\gamma_z$, all eigenvalues have nonpositive real parts, and Eq.~\eqref{eqn:quasi_bloch} has a unique fixed point which is the maximally mixed state. Note that the Hamiltonian-induced stability is contingent upon both the form of the Hamiltonian (in this case, choosing $\hat{H} \propto \hat{\sigma}_y$) and the parameter choices. This dependency also persists in the more complex examples discussed below.
A more general formulation of the Hamiltonian-induced stability may be captured by the hypocoercivity theory, which was originally formulated in the context of classical kinetic theory~\cite{Villani2007} and has recently been applied in the context of open quantum systems described by Lindblad equations~\cite{fang2024mixingtimeopenquantum}.

Next, we go beyond single spin dynamics and illustrate the form of Hamiltonian-induced stability arising in a noninteracting fermionic impurity model, assuming a system-only Hamiltonian $\hat{H}_S = \sum_{jj'} (\bm{H}_S)_{jj'} \hat{a}_j^\dag \hat{a}_{j'}$. 
The noninteracting nature of this fermionic model allows us to describe its dynamics using a Gaussian fermionic formalism, where we formulate the equation of motion in terms of a one-particle reduced density matrix (1-RDM) $P_{pq}(t) = \tr[\hat{c}_q^\dag \hat{c}_p \hat{\rho}_{SA}(t)]$. The equation of motion for $P_{pq}(t)$
is given by the continuous-time differential Lyapunov equation~\cite{Lotem2020, Barthel_2022}:
\begin{equation}
    \partial_t \bm{P} = \bm{X}\bm{P} + \bm{P}\bm{X}^\dag + \bm{Y},
    \label{eq:lyapunov_rdm}
\end{equation}
where
\begin{equation}
    \begin{gathered}
        \bm{X} = \left(
        \begin{array}{ccc}
           -i \boldsymbol{H}_{S}  & -i\bm{V}_1-\bm{M}_1 &  -i\bm{V}_2-\bm{M}_2 \\
           -i\bm{V}_1^\dag - \bm{M}_1^\dag  & -\bm{Z}_1 & \bm{0} \\
            -i\bm{V}_2^\dag - \bm{M}_2^\dag & \bm{0} & -\bm{Z}_2
        \end{array}
        \right), \\
        \bm{Y} = \left(
        \begin{array}{ccc}
           \bm{0}  & \bm{0} &  2\bm{M}_2 \\
           \bm{0}  & \bm{0} & \bm{0} \\
            2\bm{M}_2^\dag & \bm{0} & 2 \bm{\Gamma}_2
        \end{array}
    \right).
    \end{gathered}
\end{equation}
The Lyapunov equation in Eq.~\eqref{eq:lyapunov_rdm} is asymptotically stable if and only if the real parts of all eigenvalues of $\bm{X}$ are strictly negative~\cite{Simoncini2016}.  For simplicity, we now assume a single-site system, i.e., $\hat{H}_S = h \hat{a}^\dag \hat{a}$ where $h \in\mathbb R$, and $\bm{X}$ is
\begin{equation}
    \bm{X} = \left( 
    \begin{array}{cc}
         -i h &  -i\bm{V}-\bm{M} \\
         -i\bm{V}^\dag - \bm{M}^\dag & -\bm{Z}
    \end{array}
    \right),
    \label{eq:single_site_X}
\end{equation}
where we do not separate baths 1 and 2 since there is no difference in the expression for $\bm{X}$. 
In this setting, and assuming without loss of generality,  transformation to a diagonal $\mathbf{Z}$ with elements $z_k = \gamma_k + iE_k$ (see discussion around Eq.~(\ref{eq:Corrt_multisite})) we can prove that when $\|\bm{M}\|<\frac{1}{2}\min\gamma_{k}$, if the magnitude of the Hamiltonian term $\|\bm{V}\|$ is sufficiently large, then the dynamics in Eq.~\ref{eq:lyapunov_rdm} with $\bm{X}$ being Eq.~\ref{eq:single_site_X} is asymptotically stable (see the proof in the SM~\cite{supp_mat}).

The noninteracting fermionic impurity model has a special feature that the eigenvalue spectrum of $\bm{X}$, which determines the asymptotic stability of dynamics, is dependent on the weight $w_k = (V_k - i M_k)(V_k + i M_k)^*$, but not on the gauge choice for each $V_k$ and $M_k$. 
In contrast, in non-integrable systems, such as the ones studied earlier in this work, we observe different stability behavior depending on the gauge degrees of freedom. Take a single-mode spin-boson problem, for instance. In HEOM simulations, this problem has been observed to have instabilities~\cite{Krug2023} after a bosonic Hilbert space truncation. The problem with the reported instability is defined by $\hat{H}_S = \Delta  \hat{\sigma}_x/2$, $\hat{S}_{j=1}=\hat{\sigma}_z$, and the BCF $C(t) = we^{-zt}$, with $\Delta = 4.0$, $w=50.0-2.5i$, and $z=1.0$, and the initial reduced density operator is $\hat{\rho}_S(0) = \proj{1}$. The Hamiltonian and quasi-Lindblad parameters, $V, M \in \mathbb{C}$, satisfy the constraint $(V - iM)(V + iM)^* = w$. We vary $V$ and $M$ through the parameterization $V - iM = \kappa \sqrt{w}, V + iM = \kappa^{-1} \sqrt{w^*}$. Fig.~\ref{fig:instability} shows the time evolution of the excited state occupation, $P_1(t) = \bra{1} \hat{\rho}_S(t) \ket{1}$ at $\kappa=1$, $\kappa=2$, and $\kappa=0.5$. We recall that $\kappa=1$ is the parameter that minimizes $|M|$. 
We impose a Hilbert space truncation by limiting the maximum number of bosonic excitations, $n \leq N_{\max}$. At $\kappa=1$, the results converge with increasing $N_{\max}$ without any propagation instabilities, whereas at $\kappa=2$ and $\kappa=0.5$, the results show instabilities in the long time limit, similar to the HEOM results in \cite{Krug2023}.

Figs.~\ref{fig:instability} (d--f) show the eigenvalues $\lambda$ of the Liouvillian $\mathcal{L}= \mathcal{H}+\mathcal{D}$ with $N_{\max}=40$. The maximum real part of the eigenvalues is 0 at $\kappa=1$, but it is positive at $\kappa=2$ and $\kappa=0.5$, which shows the origin of the differing stability behavior. In Fig.~\ref{fig:instability_phase}, the maximum real part of the eigenvalues is plotted between $\kappa=1$ and $\kappa=2$ (Fig.~\ref{fig:instability_phase}a) and between $1/\kappa=1$ and $1/\kappa=2$ (Fig.~\ref{fig:instability_phase}b). This clearly shows a transition from stable to unstable dynamics at $\kappa > 1.3$ and $1/\kappa>1.2$.

\begin{figure}[t]
    \centering
    \includegraphics[width=1.0\columnwidth]{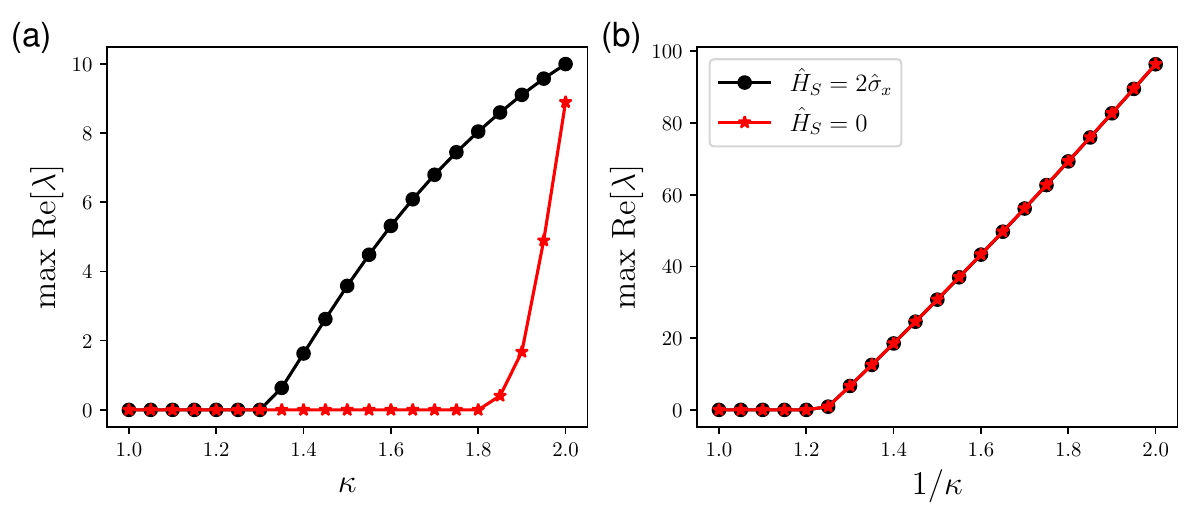}
    \caption{The maximum real parts of eigenvalues of $\mathcal{L}$ as a function of $\kappa$ with $\hat{H}_S = 2 \hat{\sigma}_x$ (black dots) and $\hat{H}_S = 0$ (red stars). $\kappa=1$ is the recommended choice in this work and is found to be numerically stable.}
    \label{fig:instability_phase}
\end{figure}

This result shows that our choice $\kappa=1$ that minimizes $|M|$ provides a stable solution for this BCF. However, we are not yet aware of a parameterization that can guarantee the stability of the dynamics for an arbitrary BCF, and the verification becomes more challenging in general multi-mode cases.

A practical recipe is to guess $\kappa$ from a related problem where the stability can be easily determined. Here, we propose a feasible procedure for a multi-mode boson model.
In particular, the spin-boson model has an analytical solution when the system Hamiltonian $\hat{H}_S$ commutes with the coupling operator $\hat{S}$. For simplicity, we take $\hat{H}_S=0$, with a single spin-site, and the limit of a single-coupling. Then, the system-reduced density operator in the coupling operator eigenbasis, $\hat{\rho}_S = \ket{s}\!\!\bra{s'}$ and $\hat{S}\ket{s} = s \ket{s}$, the density operator $\hat{\rho} = \hat{\rho}_S \otimes \hat{\rho}_B$ forms an invariant subspace of the Liouvillian $\mathcal{L}$,
\begin{equation}
    \mathcal{L}(\hat{\rho}_S \otimes \hat{\rho}_B) = \hat{\rho}_S \otimes \mathcal{L}_{s,s'}(\hat{\rho}_B),
\end{equation}
where $\mathcal{L}_{s,s'}$ is an effective Liouvillian within the auxiliary bath depending on $\hat{\rho}_S$ ($s$ and $s'$). $\mathcal{L}_{s,s'}$ is explicitly expressed as,
\begin{align}
    \mathcal{L}_{s,s'} \bullet &= \mathcal{D}_A \bullet  -is \hat{A} \bullet +is' \bullet \hat{A} \nonumber \\
    &+\sum_k \left((2s'-s) M_k \hat{d}_k - s M_k^* \hat{d}_k^\dag \right) \bullet \nonumber \\
    & + \sum_k \bullet \left(-s' M_k \hat{d}_k + (2s-s') M_k^* \hat{d}_k^\dag  \right),
\end{align}
with $\hat{A} = \sum_k V_k \hat{d}_k + V_k^* \hat{d}_k^\dag$. We note that if $\mathcal{D}_A$ is given in the diagonal basis (diagonal $\bm{Z}$ in Eq.~\ref{eq:Corrt_multisite}), the Liouvillian becomes a sum of independent Liouvillians, $\mathcal{L}_{s,s'} = \sum_k \mathcal{L}_{s,s',k}$. In this case, we can determine the stability of multi-mode spin-boson problems from single-mode spin-boson calculations with $\mathcal{L}_{s,s',k}$.

Although the above calculation is only applicable in the analytically solvable case, we can imagine that the stability is still related to that under general $\hat{H}_S$ dynamics.
To assess this correlation, we compute the maximum real part of the eigenvalues of $\mathcal{L}$ as a function of $\kappa$ with $\hat{H}_S=0$, and with the other parameters set to be the same as in the previous setup in Fig.~\ref{fig:instability_phase}. This calculation is done by computing the eigenvalues of $\mathcal{L}_{s,s'}$ (single $k$ in this case) with $s,s' = \pm 1$ and taking the maximum value over them. We also observe that all the eigenvalues with positive real parts are from $s=1, s'=-1$, implying that the origin of instability is from the nondiagonal elements of the system density operator. For $\kappa < 1$, the computed eigenvalues match remarkably well between the analytically tractable $\hat{H}_S=0$ dynamics and the case of $\hat{H}_S=2\hat{\sigma}_x$ considered previously. 
For $\kappa > 1$, the eigenvalues overestimate the region of stability in $\kappa$ but nonetheless correctly predicts that a region around $\kappa=1$ is stable.

\section{Conclusions}
We have presented the quasi-Lindblad pseudomode theory, which provides a flexible framework to incorporate a highly compact representation of the bath correlation function (BCF) of an open quantum system. 
Specifically, the BCF in the quasi-Lindblad pseudomode theory takes the form of a sum of complex exponential functions with complex weights, in contrast to positive weights 
in the conventional Lindblad pseudomode theory. By adopting an efficient numerical algorithm for the complex exponential function fitting procedure, we showed that the number of pseudomodes can be significantly reduced, and only a small number of exponential functions ($<10$) is sufficient for both fitting the BCF and describing the associated reduced system dynamics.

The quasi-Lindblad pseudomode theory achieves this representation of the BCF by relaxing the complete positivity (CP) of the global dynamics consisting of the system and the pseudomodes. Although the total system-bath density operator along the dynamics remains trace-preserving throughout the dynamics, it is generally not positive.  Nonetheless, after tracing out the bath, the quasi-Lindblad pseudomode theory can still produce an accurate system-reduced density operator, which can be very close to a positive operator, provided the error in the BCF relative to the original unitary problem is sufficiently small. The dependence of the error of the system-reduced density operator on the BCF has been rigorously analyzed in the context of unitary dynamics for spin-boson systems~\cite{Mascherpa2017, vilkoviskiy2023bound, liu2024errorboundsopenquantum}, and work is in progress towards extending such results for fermionic systems and for the Lindblad and quasi-Lindblad theories~\cite{HuangPark2024open}.

The pseudomode formulation provides an alternative framework to other numerical approaches to open quantum system dynamics, which have also used a complex-weighted sum of exponential functions to represent the BCF, such as the HEOM~\cite{Stockburger2022, DanShi2023, xu2023universal}. An important strength is that once the pseudomodes are determined, there exists a wide variety of numerical methods that can be used to propagate the resulting master equation~\cite{Somoza2019, Lotem2020, Brenes2020, Schwarz2016, Dorda2014, Dorda2015, WernerArrigoni2023, werner2023auxiliary, LuoCirio2023prxquantum}.

A crucial theoretical question raised by the quasi-Lindblad formulation is the stability of the global dynamics.
Despite relaxing the CP condition, we found that the combination of the Hamiltonian and dissipator terms can lead to dynamics which is stable, a phenomenon we refer to as  Hamiltonian-induced stability. 
Further, because the quasi-Lindblad formulation contains a set of gauge degrees of freedom, we can take advantage of them to create a stable global dynamics for a particular system, as we demonstrated in the 
case of the single-mode spin-boson model.
In this example, adjusting the gauge parameters allowed us to resolve an instability issue that is known to occur when using the HEOM method~\cite{Krug2023}. Nonetheless, although such instabilities are a formal possibility within our theory and should be studied further, the widespread use of HEOM in open quantum system simulations also suggests that they may not be relevant to a range of physical studies.

We anticipate that this theoretical development can facilitate accurate simulations of open quantum systems within more complex environments, such as those encountered in photosynthetic systems~\cite{fleming2009}, nonequilibrium Kondo problems~\cite{Hewson1993}, and quantum control~\cite{Nori2023qutip}.

\vspace{1em}
\emph{Acknowledgements.--}  This
work is an equal collaboration between two SciDAC teams ``Real-time dynamics of driven correlated electrons in quantum materials'' and ``Traversing the `death valley' separating short and long times in non-equilibrium quantum dynamical simulations of real materials'', supported by the U.S. Department of Energy, Office of Science, Office of Advanced Scientific Computing Research and Office of Basic Energy Sciences, Scientific Discovery through Advanced Computing (SciDAC) program under Award Number DE--SC0022088 (G.P., C.Y., G.K.C.) and DE--SC0022198 (L.L., C.Y., Y.Z.). The work by Z.H. was supported by the Simons Targeted Grants in
Mathematics and Physical Sciences on Moir\'e Materials Magic.
G.K.C. and L.L. are Simons Investigators.  We thank Zhiyan Ding, Joonho Lee, David Limmer, Vojt\v{e}ch Vl\v{c}ek, Erika Ye, Lexing Ying, and Neill Lambert for helpful discussions.

\vspace{1em}
\emph{Author contributions.--} G.P., Z.H., G.K.C., L.L. conceived the original study.
G.P., Z.H., Y.Z., and L.L. carried out theoretical analysis to support the study. G.P. and Z.H. carried out numerical calculations to support the study. All authors, G.P., Z.H., Y.Z., C.Y., G.K.C., and L.L., discussed the results of the manuscript and contributed to the writing of the manuscript.  

\vspace{1em}
\emph{Note added.--} During submission of this manuscript, a related paper~\cite{Thoenniss2024} appeared, which discusses analytical numerical scaling and provides different numerical fitting schemes.

\bibliography{references} 

\include{supp_mat}

\end{document}

%% file: supp_mat.tex
\newpage
\setcounter{equation}{0}
\setcounter{figure}{0}
\setcounter{table}{0}
\setcounter{page}{1}
\setcounter{section}{0}
\makeatletter
\renewcommand{\thesection}{SM-\Roman{section}}
\renewcommand{\theequation}{S\arabic{equation}}
\renewcommand{\thefigure}{S\arabic{figure}}
\renewcommand{\bibnumfmt}[1]{[S#1]}

\title{Supplemental Material for ``Quasi-Lindblad pseudomode theory for open quantum systems"}

\begin{CJK*}{UTF8}{mj}
\author{Gunhee Park (박건희)}
\affiliation{Division of Engineering and Applied Science, California Institute of Technology, Pasadena, California 91125, USA}
\author{Zhen Huang}
\affiliation{Department of Mathematics, University of California, Berkeley, California 94720, USA}
\author{Yuanran Zhu}
\affiliation{Applied Mathematics and Computational Research Division, Lawrence Berkeley National Laboratory, Berkeley, California 94720, USA}
\author{Chao Yang}
\affiliation{Applied Mathematics and Computational Research Division, Lawrence Berkeley National Laboratory, Berkeley, California 94720, USA}
\author{Garnet Kin-Lic Chan}
\affiliation{Division of Chemistry and Chemical Engineering, California Institute of Technology, Pasadena, California 91125, USA}
\author{Lin Lin}
\affiliation{Department of Mathematics, University of California, Berkeley, California 94720, USA}
\affiliation{Applied Mathematics and Computational Research Division, Lawrence Berkeley National Laboratory, Berkeley, California 94720, USA}

\maketitle
\end{CJK*}
\onecolumngrid

\section{Proof of the equivalence condition for the bosonic Gaussian bath}

In this section, we provide the derivation of the equivalence condition for the bosonic Gaussian bath based on the bath correlation function (BCF). Our strategy here is different from that in \cite{Tamascelli2018}, which is only applicable to the Lindblad equation and relies on the dilation of the Lindblad system into an infinite dimensional unitary system.

We start from the following system-bath Liouvillian superoperator form introduced in the main text,
\begin{equation}
    \mathcal{L}_{SA} = -i \sum_j \mathcal{S}_j \mathcal{F}_j +i \sum_j \widetilde{\mathcal{S}}_j \widetilde{\mathcal{F}}_j.
\end{equation}
For example, in the case when there is only Hamiltonian coupling, $\hat{H}_{SA} = \sum_j \hat{S}_j \hat{A}_j$, $\mathcal{F}_j\bullet = \hat{A}_j \bullet$ and $\widetilde{\mathcal{F}}_j \bullet = \bullet \hat{A}_j$. In the case of the quasi-Lindblad pseudomode with both Hamiltonian and Lindblad couplings with $\mathcal{D}_{SA} \bullet = \sum_{j} \hat{L}'_{j} \bullet \hat{S}_j + \hat{S}_j \bullet \hat{L}'^\dag_{j} - \frac{1}{2}\{ \hat{S}_j \hat{L}'_{j} + \hat{L}'^\dag_{j}  \hat{S}_j, \bullet  \}$,
\begin{equation}
    \begin{gathered}
        \mathcal{F}_j \bullet = \hat{A}_j \bullet +  \bullet i \hat{L}'^\dag_{j} - \frac{i}{2}  (\hat{L}'_{j} + \hat{L}'^\dag_{j} ) \bullet, \\
        \widetilde{\mathcal{F}}_j \bullet = \bullet \hat{A}_j - i \hat{L}'_{j} \bullet +  \bullet \frac{i}{2} (\hat{L}'_{j} + \hat{L}'^\dag_{j} ).
    \end{gathered}
    \label{smeq:F_superoperator}
\end{equation}
For brevity, we will not explicitly distinguish between the
non-tilde and tilde superoperators, and denote  $\mathcal{L}_{SA} = -i \sum_j \mathcal{S}_j \mathcal{F}_j$, treating $\widetilde{\mathcal{S}}_j = \mathcal{S}_{j_{\max}+j}$ and $\widetilde{\mathcal{F}}_j =- \mathcal{F}_{j_{\max}+j}$. We will return to the explicit tilde superoperator notation when the distinction is necessary.

In the interaction picture, $\hat{\widetilde{\rho}}_{SA}(t) = e^{-(\mathcal{L}_S + \mathcal{L}_A)t} \hat{\rho}_{SA}(t) $, $\mathcal{S}_j(t) = e^{-\mathcal{L}_S t} \mathcal{S}_j e^{\mathcal{L}_S t}$, $\mathcal{F}_j(t) = e^{-\mathcal{L}_A t} \mathcal{F}_j e^{\mathcal{L}_A t} $, and $\mathcal{L}_{SA}(t) = -i \sum_j \mathcal{S}_j(t) \mathcal{F}_j(t) $, the von Neumann equation of motion for $\hat{\widetilde{\rho}}_{SA}(t)$ is,
\begin{equation}
    \partial_t \hat{\widetilde{\rho}}_{SA}(t) = \mathcal{L}_{SA}(t) \hat{\widetilde{\rho}}_{SA}(t).
\end{equation}
Note that the backward evolution operators such as $e^{-\mathcal{L}_A t}$ are introduced mainly to simplify the notation, and these operators do not appear explicitly after applying the time-ordering operations below.
The formal solution of $\hat{\widetilde{\rho}}_{SA}(t)$ can be obtained with the Dyson series expansion~\cite{rivas2012open},
\begin{equation}
    \hat{\widetilde{\rho}}_{SA}(t) = \sum_{m=0}^\infty \frac{1}{m!} \mathcal{T} \left( \int_0^t \cdots \int_0^t \mathcal{L}_{SA}(t_1) \cdots \mathcal{L}_{SA}(t_m) dt_1 \cdots dt_m \right) \hat{\rho}_{SA}(0),
\end{equation}
where $\mathcal{T}$ refers to the time-ordering superoperator. We trace out the bath to obtain the system-reduced density operator,
\begin{equation}\label{smeq:dyson_rhoS1}
    \hat{\widetilde{\rho}}_S(t) = \sum_{m=0}^\infty \frac{(-i)^m}{m!} \int_0^t \cdots \int_0^t \sum_{j_1} \cdots \sum_{j_m} \tr_A \left[ \mathcal{T}_A (\mathcal{F}_{j_1}(t_1) \cdots \mathcal{F}_{j_m}(t_m)) \hat{\rho}_A(0) \right] \mathcal{T}_S (\mathcal{S}_{j_1}(t_1) \cdots \mathcal{S}_{j_m}(t_m)) \hat{\rho}_S(0).
\end{equation}
We separate the time-ordering superoperator into the system time-ordering superoperator $\mathcal{T}_S$ and the bath time-ordering superoperator $\mathcal{T}_A$ to be applied within each Liouville operator space. We denote $C_{j_1,\cdots,j_m}(t_1,\cdots,t_m) = \tr_A \left[ \mathcal{T}_A (\mathcal{F}_{j_1}(t_1) \cdots \mathcal{F}_{j_m}(t_m)) \hat{\rho}_A(0) \right] $. From Wick's theorem, if $m$ is odd, it is zero, and if $m$ is even ($m=2n$),
\begin{equation}
    C_{j_1,\cdots,j_{2n}}(t_1,\cdots,t_{2n}) = \sum_{\sigma \in \Pi_{2n}} \prod_{i=1}^n C_{j_{\sigma(2i-1)}, j_{\sigma(2i)}}(t_{\sigma(2i-1)},t_{\sigma(2i)}).
\end{equation}
After inserting it into Eq.~\ref{smeq:dyson_rhoS1},
\begin{equation}\label{smeq:dyson_rhoS2}
    \hat{\widetilde{\rho}}_S(t) = \sum_{n=0}^\infty \frac{(-1)^n}{(2n)!} \int_0^t \cdots \int_0^t \sum_{j_1} \cdots \sum_{j_n} \sum_{\sigma \in \Pi_{2n}}  \mathcal{T}_S \left( \prod_{i=1}^n C_{j_{\sigma(2i-1)}, j_{\sigma(2i)}}(t_{\sigma(2i-1)},t_{\sigma(2i)}) \mathcal{S}_{j_1}(t_1) \cdots \mathcal{S}_{j_n}(t_n)
 \right)\hat{\rho}_S(0) dt_1 \cdots dt_n
\end{equation}
Thanks to the system time ordering superoperator, $\mathcal{T}_S$, $\mathcal{T}_S \left( \prod_{i=1}^n C_{j_{\sigma(2i-1)}, j_{\sigma(2i)}}(t_{\sigma(2i-1)},t_{\sigma(2i)}) \mathcal{S}_{j_1}(t_1) \cdots \mathcal{S}_{j_n}(t_n) \right)$ becomes the same operator for all $\sigma \in \Pi_{2n}$. The number of elements of $\Pi_{2n}$ is $(2n)!!$, so Eq.~\ref{smeq:dyson_rhoS2} is reduced to,
\begin{align}
    \hat{\widetilde{\rho}}_S(t) &= \sum_{n=0}^\infty \frac{(-1)^n}{2^n n!} \int_0^t \cdots \int_0^t \sum_{j_1} \cdots \sum_{j_n} \mathcal{T}_S \left( \prod_{i=1}^n C_{j_{2i-1}, j_{2i}}(t_{2i-1},t_{2i}) \mathcal{S}_{j_1}(t_1) \cdots \mathcal{S}_{j_n}(t_n)
    \right)\hat{\rho}_S(0) dt_1 \cdots dt_n \\
    &= \mathcal{T}_S \exp \left( -\frac{1}{2} \int_0^t \int_0^t \sum_{j_1, j_2} C_{j_1,j_2}(t_1,t_2) \mathcal{S}_{j_1}(t_1) \mathcal{S}_{j_2}(t_2) dt_1 dt_2 \right) \hat{\rho}_S(0) \\
    &= \mathcal{T}_S \exp \left( -\int_0^t \int_0^{t_1} \sum_{j_1, j_2} C_{j_1,j_2}(t_1,t_2) \mathcal{S}_{j_1}(t_1) \mathcal{S}_{j_2}(t_2) dt_1 dt_2 \right) \hat{\rho}_S(0).
\end{align}
From the second to the third line, we fix the ordering between $t_1$ and $t_2$ as $t_1>t_2$. After returning to the original frame and parameterizing the integrand in the time non-local exponential superoperator,
\begin{equation}
    \hat{\rho}_S(t) = e^{\mathcal{L}_S t} \mathcal{T}_S \exp \left( \int_0^t \int_0^{t_1} \mathcal{W}(t_1,t_2) dt_1 dt_2 \right) \hat{\rho}_S(0),
\end{equation}
which is termed  the influence superoperator~\cite{breuer2007open, Aurell_2020, Cirio2022}. The expression for $\mathcal{W}(t_1,t_2)$ after reviving the tilde superoperators becomes,
\begin{align}
    \mathcal{W}(t_1,t_2) = &- \sum_{jj'} \Big( \langle \mathcal{F}_j(t_1) \mathcal{F}_{j'}(t_2) \rangle_A \mathcal{S}_j(t_1) \mathcal{S}_{j'}(t_2) - \langle \mathcal{F}_j(t_1) \widetilde{\mathcal{F}}_{j'}(t_2) \rangle_A \mathcal{S}_j(t_1) \widetilde{\mathcal{S}}_{j'}(t_2) \nonumber \\
    &- \langle \widetilde{\mathcal{F}}_j(t_1) \mathcal{F}_{j'}(t_2) \rangle_A \widetilde{\mathcal{S}}_j(t_1) \mathcal{S}_{j'}(t_2) + \langle \widetilde{\mathcal{F}}_j(t_1) \widetilde{\mathcal{F}}_{j'}(t_2) \rangle_A \widetilde{\mathcal{S}}_j(t_1) \widetilde{\mathcal{S}}_{j'}(t_2) \Big),
\end{align}
where $\langle \bullet \rangle_A = \tr_A[\bullet \hat{\rho}_A(0)]$. This expression can be further simplified by using the following property, $\tr_A[\mathcal{F}_j \bullet] = \tr_A[\widetilde{\mathcal{F}}_j \bullet]$. For the Hamiltonian coupling part, it is clear to see  $\tr_A[\hat{A}_j \bullet] = \tr_A[\bullet \hat{A}_j]$. For the Lindblad coupling part in Eq.~\ref{smeq:F_superoperator}, it can be seen from  direct calculation,
\begin{equation}\label{smeq:F_trace_properties}
    \tr_A[i \bullet \hat{L}'^\dag_{j} - \frac{i}{2} (\hat{L}'_{j} + \hat{L}'^\dag_{j} )\bullet] = \tr_A[\frac{i}{2} (\hat{L}'^\dag_{j}-\hat{L}'_{j}) \bullet] = \tr_A[-i\hat{L}'_{j}\bullet + \frac{i}{2}\bullet (\hat{L}'_{j} + \hat{L}'^\dag_{j} ) ].  
\end{equation}
Combining this property with the trace-preserving property, $\tr_A[\mathcal{L}_A \bullet] = 0$ or $\tr_A[e^{-\mathcal{L}_A t} \bullet] = \tr_A[\bullet]$, we have,
\begin{equation}
    \langle \mathcal{F}_j(t_1) \mathcal{F}_{j'}(t_2) \rangle_A = \langle \widetilde{\mathcal{F}}_j(t_1) \mathcal{F}_{j'}(t_2) \rangle_A, \quad  \langle \mathcal{F}_j(t_1) \widetilde{\mathcal{F}}_{j'}(t_2) \rangle_A = \langle \widetilde{\mathcal{F}}_j(t_1) \widetilde{\mathcal{F}}_{j'}(t_2) \rangle_A.
\end{equation}
Furthermore, by using the fact that $(\mathcal{F}_j \bullet )^\dag = \widetilde{\mathcal{F}}_j (\bullet)^\dag$, we have,
\begin{equation}
    \langle \widetilde{\mathcal{F}}_j(t_1) \widetilde{\mathcal{F}}_{j'}(t_2) \rangle_A = \langle \mathcal{F}_j(t_1) \mathcal{F}_{j'}(t_2) \rangle_A^*.
\end{equation}
Therefore, by defining the BCF as,
\begin{equation}
    C_{jj'}(t_1,t_2) =  \langle \mathcal{F}_j(t_1) \mathcal{F}_{j'}(t_2) \rangle_A = \tr_A[\mathcal{F}_j e^{\mathcal{L}_A(t_1-t_2)}\mathcal{F}_{j'}e^{\mathcal{L}_A(t_2)} \hat{\rho}_A(0)],
\end{equation}
the superoperator $\mathcal{W}(t_1,t_2)$ is expressed as,
\begin{equation}
    \begin{gathered}
        \mathcal{W}(t_1,t_2) = - \sum_{jj'} \left(\mathcal{S}_j(t_1) - \widetilde{\mathcal{S}}_j(t_1)\right) \times \\
        \left(C_{jj'}(t_1,t_2) \mathcal{S}_{j'}(t_2) - C^{*}_{jj'}(t_1,t_2) \widetilde{\mathcal{S}}_{j'}(t_2) \right).
    \end{gathered}
\end{equation}
We see, therefore, that the influence of the bath on the system dynamics, as expressed through the influence superoperator, is entirely parametrized by the BCF. Thus, any two baths with the same BCF give rise to the same system dynamics, which is the equivalence condition.

\section{BCF expression for the spin-boson model}\label{smsec:sec3}

In this section, we derive the BCF expression for the spin-boson model,
\begin{equation}
    \bm{C}^A(\Delta t) = (\bm{V}-i\bm{M}) e^{-\bm{Z}\Delta t} (\bm{V}+i\bm{M})^\dag.
    \label{smeq:Corrt_multisite}
\end{equation}
This can be shown from a direct calculation. First, $e^{\mathcal{L}_A t'} \hat{\rho}_A(0) = e^{\mathcal{L}_A t'} \proj{\bm{0}} = \proj{\bm{0}}$ from $\mathcal{L}_A \proj{\bm{0}} = 0$. We have,
\begin{equation}
    \mathcal{F}_{j'} \proj{\bm{0}} = \sum_k (V_{j'k}^* - i M_{j'k}^*) \hat{d}_k^\dag \proj{\bm{0}},
\end{equation}
from $\hat{d}_k \proj{\bm{0}} = \proj{\bm{0}} \hat{d}_k^\dag = 0$. Then, from $\mathcal{L}_A[\hat{d}_k^\dag \proj{\bm{0}}] = \sum_l (-iH^A_{lk}  - \Gamma_{lk})\hat{d}_l^\dag \proj{\bm{0}}$, we get $e^{\mathcal{L}_A t} \left[\sum_k (V_{j'k}^* - i M_{j'k}^*) \hat{d}_k^\dag \proj{\bm{0}}\right] = \sum_{lk} (e^{-\bm{Z}t})_{lk} (V_{j'k}^* - i M_{j'k}^*) \hat{d}_l^\dag \proj{\bm{0}}$ where $\bm{Z}=\bm{\Gamma}+i\bm{H}^A$. Using the property in Eq.~\ref{smeq:F_trace_properties},
\begin{equation}
    \tr_A \left[ \mathcal{F}_j \bullet \right] = \tr_A \left[ (\sum_k (V_{jk} -i M_{jk}) \hat{d}_k + (V^*_{jk} +i M^*_{jk}) \hat{d}^\dag_k ) \bullet \right].
\end{equation}
Therefore, this leads to,
\begin{equation}
    C_{jj'}^A(t+t',t') = \tr_A \left[ \mathcal{F}_j e^{\mathcal{L}_A t} \mathcal{F}_{j'} e^{\mathcal{L}_A t'} \hat{\rho}_A(0) \right] = \sum_{kl} (V_{jl} - iM_{jl}) (e^{-\bm{Z}t})_{lk} (V_{j'k}^* - i M_{j'k}^*),
\end{equation}
or in the matrix representation,
\begin{equation}
    \bm{C}^A(t) = (\bm{V}-i\bm{M}) e^{-\bm{Z} t} (\bm{V}+i\bm{M})^\dag.
\end{equation}

\section{BCF expression for the fermionic impurity model}
In this section, we discuss the quasi-Lindblad pseudomode formulation for the fermionic impurity model.
The pseudomode model contains two different auxiliary baths, $A_1$ and $A_2$ with the initial state, $\hat{\rho}_{SA}(0) = \hat{\rho}_S(0) \otimes \hat{\rho}_A(0)$ with $\hat{\rho}_A(0) = \hat{\rho}_{A_1}(0) \otimes \hat{\rho}_{A_2}(0) = \proj{\bm{0}} \otimes \proj{\bm{1}}$. We define index sets, $\mathbb{I}_1$ and $\mathbb{I}_2$, to denote the auxiliary mode index $k_1 \in \mathbb{I}_1$ ($k_2 \in \mathbb{I}_2$) for the fermions in the auxiliary bath $A_1$ ($A_2$), respectively. The density operator evolves under the following Lindblad equation,
\begin{equation}
    \partial_t \hat{\rho}_{SA} = -i [\hat{H}_\aux, \hat{\rho}_{SA} ] + (\mathcal{D}_{A_1} + \mathcal{D}_{A_2} + \mathcal{D}_{SA_1} + \mathcal{D}_{SA_2}) \hat{\rho}_{SA}. \label{smeq:pseudomode_fermion}
\end{equation}
The auxiliary Hamiltonian is given by,
\begin{equation}
    \hat{H}_\aux = \hat{H}_S +\sum_{k \in \mathbb{I}_1 \cup \mathbb{I}_2 } E_k \hat{d}^\dag_k \hat{d}_k + \sum_j \hat{a}^\dag_j \hat{A}_j + \hat{A}_j^\dag \hat{a}_j, \quad \hat{A}_j = \sum_{k \in \mathbb{I}_1 } (\bm{V}_1)_{jk_1} \hat{d}_{k_1} + \sum_{k \in \mathbb{I}_2 } (\bm{V}_2)_{jk_2} \hat{d}_{k_2}.
    \label{smeq:pseudomode_fermion2}
\end{equation}
The dissipators are given by,
\begin{equation}
    \begin{gathered}
        \mathcal{D}_{A_1} \bullet = \sum_{k_1 \in \mathbb{I}_1} 2 \gamma_{k_1} \left( \hat{d}_{k_1} \bullet \hat{d}_{k_1}^\dag - \frac{1}{2} \{ \hat{d}_{k_1}^\dag \hat{d}_{k_1}, \bullet \} \right), \\
        \mathcal{D}_{A_2} \bullet = \sum_{k_2 \in \mathbb{I}_2} 2 \gamma_{k_2} \left( \hat{d}^\dag_{k_2} \bullet \hat{d}_{k_2} - \frac{1}{2} \{ \hat{d}_{k_2} \hat{d}_{k_2}^\dag, \bullet \} \right), \\
        \mathcal{D}_{SA_1} \bullet  = \sum_j   \hat{a}_j \bullet \hat{L}_{1j}^\dag + \hat{L}_{1j} \bullet \hat{a}_j^\dag - \frac{1}{2} \{\hat{L}^\dag_{1j} \hat{a}_j + \hat{a}_j^\dag \hat{L}_{1j}, \bullet \},\\
        \mathcal{D}_{SA_2} \bullet = \sum_j   \hat{a}^\dag_j \bullet \hat{L}_{2j} + \hat{L}^\dag_{2j} \bullet \hat{a}_j - \frac{1}{2} \{\hat{L}_{2j} \hat{a}^\dag_j + \hat{a}_j \hat{L}^\dag_{2j}, \bullet \},
    \end{gathered}
    \label{smeq:pseudomode_fermion3}
\end{equation}
with $\hat{L}_{1j} = 2 \sum_{k_1 \in \mathbb{I}_1} (\bm{M}_1)_{jk_1} \hat{d}_{k_1} $ and $\hat{L}_{2j} = 2 \sum_{k_2 \in \mathbb{I}_2} (\bm{M}_2)_{jk_2} \hat{d}_{k_2} $.

We adopt a super-fermion representation~\cite{Schmutz1978, kosov2011, Dorda2014, Park2024TNIF} of Lindblad dynamics to formulate the superoperator formalism for the fermionic system. This representation maps a density operator to a state vector, $\hat{\rho} \mapsto \kett{\rho} $, in the so-called `super-Fock' space with twice the number of original orbitals. This is done with a left-vacuum vector,
\begin{equation}
    \kett{I} = \prod_{k} (\hat{d}_k^{\dag}+\hat{\Tilde{d}}_k^{\dag})\ket{\bm{0}} \otimes \ket{\bm{0}},
\end{equation}
where $\Tilde{d}_k^\dag$ is a fermionic creation operator in the `tilde' space. The density operator is then represented as $\kett{\rho} = (\hat{\rho} \otimes \mathbb{I}) \kett{I} $. For example, the initial bath density operator, $\hat{\rho}_A(0) = \proj{\bm{0}} \otimes \proj{\bm{1}}$ becomes,
\begin{equation}
    \hat{\rho}_A(0) \mapsto \kett{\rho_A(0)} = \kett{01}^{\otimes \mathbb{I}_1} \otimes \kett{10}^{\otimes \mathbb{I}_2}.
\end{equation}
The trace operation can also be expressed with the left-vacuum vector,
\begin{equation}
    \tr[\hat{O} \hat{\rho} ] = \braa{I} \hat{O}\otimes \mathbb{I} \kett{\rho}.
\end{equation}
The following conjugation rules for the left-vacuum vector will become useful,
\begin{equation}\label{smeq:conjugation}
    \hat{d}_k \kett{I} =  \hat{\Tilde{d}}_k \kett{I}, \quad \hat{d}_k^\dag = - \hat{\Tilde{d}}_k^\dag \kett{I}.
\end{equation}
We will assume that the initial density operator is an even-parity density operator for simplicity.
Any superoperator $\mathcal{L}$ can be cast to the operator in the super-Fock space, $\hat{\mathcal{L}}$, or $\mathcal{L} \hat{\rho} \mapsto \hat{\mathcal{L}}\kett{\rho}$. For example, the following dissipator maps to,
\begin{equation}
    \begin{gathered}
        \mathcal{D}\bullet = 2\gamma_k \left(\hat{d}_k \bullet \hat{d}_k^\dag - \frac{1}{2}\{\hat{d}_k^\dag \hat{d}_k, \bullet\} \right) \mapsto
        \hat{\mathcal{D}} = 2\gamma_k \hat{\Tilde{d}}_k^\dag \hat{d}_k - \gamma_k \hat{d}_k^\dag \hat{d}_k + \gamma_k \hat{\Tilde{d}}_k^\dag \hat{\Tilde{d}}_k - \gamma_k.
    \end{gathered}
\end{equation}
This relation can be shown from the direct calculation of $(\mathcal{D} \hat{\rho} ) \kett{I} = \hat{\mathcal{D}} \kett{\rho}$ with the conjugation rule in Eq.~\ref{smeq:conjugation}. The bath Liouvillian superoperator is written as,
\begin{align}
    \hat{\mathcal{L}}_A &= \sum_{k_1 \in \mathbb{I}_1 } (-iE_{k_1} -\gamma_{k_1} ) \hat{d}^\dag_{k_1} \hat{d}_{k_1} +(-i E_{k_1} + \gamma_{k_1} ) \hat{\Tilde{d}}^\dag_{k_1} \hat{\Tilde{d}}_{k_1} + 2\gamma_{k_1} \hat{\Tilde{d}}_{k_1}^\dag \hat{d}_{k_1} +iE_{k_1} - \gamma_{k_1} \nonumber \\
    &+ \sum_{k_2 \in \mathbb{I}_2 } (-iE_{k_2} +\gamma_{k_2} ) \hat{d}^\dag_{k_2} \hat{d}_{k_2} +(-i E_{k_2} - \gamma_{k_2} ) \hat{\Tilde{d}}^\dag_{k_2} \hat{\Tilde{d}}_{k_2} + 2\gamma_{k_2} \hat{{d}}_{k_2}^\dag \hat{\Tilde{d}}_{k_2} + iE_{k_2} - \gamma_{k_2}.
\end{align}
Note that $\hat{\mathcal{L}}_A \kett{\rho_A(0)} = 0$. With this representation, we decompose the coupling superoperators in the following form,
\begin{equation}
    \hat{\mathcal{L}}_{SA} = -i \sum_j \left( \hat{a}_j^\dag \hat{\mathcal{F}}^-_{j} + \hat{\mathcal{F}}^+_{j} \hat{a}_j + \hat{\Tilde{a}}_j^\dag \hat{\widetilde{\mathcal{F}}}^-_{j} + \hat{\widetilde{\mathcal{F}}}^+_{j} \hat{\Tilde{a}}_j \right) .
\end{equation}
where
\begin{equation}
    \begin{gathered}
        \hat{\mathcal{F}}^-_{j} = \hat{A}_j - \frac{i}{2} \hat{L}_{1j} +i  (\hat{\widetilde{L}}_{2j} + \frac{1}{2} \hat{L}_{2j}), \\
        \hat{\mathcal{F}}^+_{j} = \hat{A}^\dag_j +i  (\hat{\widetilde{L}}^\dag_{1j} - \frac{1}{2} \hat{L}_{1j}^\dag ) +\frac{i}{2} \hat{L}_{2j}^\dag, \\
        \hat{\widetilde{\mathcal{F}}}^-_{j} = \hat{\widetilde{A}}_j + i ( \hat{L}_{1j} + \frac{1}{2} \hat{\widetilde{L}}_{1j} ) - \frac{i}{2} \hat{\widetilde{L}}_{2j}, \\
        \hat{\widetilde{\mathcal{F}}}^+_{j} = \hat{\widetilde{A}}^\dag_j + \frac{i}{2} \hat{\widetilde{L}}^\dag_{1j} + i ( \hat{L}_{2j}^\dag - \frac{1}{2} \hat{\widetilde{L}}_{2j}^\dag),
    \end{gathered}
\end{equation}
with $\hat{\widetilde{A}}_j = \sum_{k \in \mathbb{I}_1 \cup \mathbb{I}_2} V_{jk} \hat{\Tilde{d}}_k $, $\hat{\Tilde{L}}_{1j} = 2 \sum_{k_1 \in \mathbb{I}_1} (\bm{M}_1)_{jk_1} \hat{\Tilde{d}}_{k_1} $ and $\hat{L}_{2j} = 2 \sum_{k_2 \in \mathbb{I}_2} (\bm{M}_2)_{jk_2} \hat{\Tilde{d}}_{k_2} $. Note that the $\hat{\mathcal{F}}$ operators with superscripts $-$ ($+$) are composed of a linear combination of $\hat{d}_k$ and $\hat{\Tilde{d}}_k$ ($\hat{d}_k^\dag$ and $\hat{\Tilde{d}}_k^\dag$), respectively. The associated superoperator BCFs are defined by,
\begin{align}
    C^{>A}_{jj'}(t_1,t_2) &=  \braa{I} \hat{\mathcal{F}}^-_{j} e^{\hat{\mathcal{L}}_A (t_1-t_2)} \hat{\mathcal{F}}^+_{j'} e^{\hat{\mathcal{L}}_A t_2} \kett{\rho_A(0)}, \\
    C^{<A}_{jj'}(t_1,t_2) &=  \braa{I} \hat{\mathcal{F}}^+_{j} 
e^{\hat{\mathcal{L}}_A (t_1-t_2)} \hat{\mathcal{F}}^-_{j'} e^{\hat{\mathcal{L}}_A t_2} \kett{\rho_A(0)}.
\end{align}
Explicit calculation leads to the BCF expression in terms of a sum of complex exponential functions,
\begin{equation}
    \begin{gathered}
        \bm{C}^{>A}(t+t',t') = (\bm{V}_{1}-i\bm{M}_{1}) e^{-\bm{Z}_{1}t} (\bm{V}_{1}+i\bm{M}_{1})^\dag, \\
        \bm{C}^{<A}(t+t',t') = (\bm{V}_{2}-i\bm{M}_{2})^* e^{-\bm{Z}_{2} t} (\bm{V}_{2}+i\bm{M}_{2})^T, 
    \end{gathered}
\end{equation}
where $\bm{Z}_{1}$ and $\bm{Z}_{2}$ are diagonal matrices with elements $z_{k_1} = \gamma_{k_1} + iE_{k_1}$ and $z_{k_2} = \gamma_{k_2}-iE_{k_2}$.

From $\hat{\mathcal{L}}_A \kett{\rho_A(0)}=0$, $e^{\hat{\mathcal{L}}_A t' } \kett{\rho_A(0)} = \kett{\rho_A(0)}$. We have,
\begin{equation}
    \begin{gathered}
        \hat{\mathcal{F}}^+_j \kett{\rho_A(0)}=(\hat{A}_j^\dag - \frac{i}{2} \hat{L}_{1j}^\dag ) \kett{\rho_A(0)} = \sum_{k_1 \in \mathbb{I}_1} ((\bm{V}_1)^*_{jk_1} - i (\bm{M}_1)^*_{jk_1}  ) \hat{d}_{k_1}^\dag \kett{\rho_A(0)},  \\
        \hat{\mathcal{F}}^-_j \kett{\rho_A(0)}=(\hat{A}_j + \frac{i}{2} \hat{L}_{2j} ) \kett{\rho_A(0)} = \sum_{k_2 \in \mathbb{I}_2} ((\bm{V}_2)_{jk_2} + i (\bm{M}_2)_{jk_2}  )  \hat{d}_{k_2} \kett{\rho_A(0)}.
    \end{gathered}
\end{equation}
For $k_1 \in \mathbb{I}_1$, $\hat{\mathcal{L}}_A \hat{d}_{k_1}^\dag \kett{\rho_A(0)} = (-iE_{k_1}-\gamma_{k_1}) \hat{d}_{k_1}^\dag \kett{\rho_A(0)}$, and for $k_2 \in \mathbb{I}_2$,  $\hat{\mathcal{L}}_A \hat{d}_{k_2} \kett{\rho_A(0)} = (iE_{k_2}-\gamma_{k_2}) \hat{d}_{k_2} \kett{\rho_A(0)}$. Hence,
\begin{equation}
    \begin{gathered}
        e^{\hat{\mathcal{L}}_At} \hat{\mathcal{F}}_j^+ \kett{\rho_A(0)} = \sum_{k_1 \in \mathbb{I}_1}  e^{(-iE_{k_1}-\gamma_{k_1})t} ((\bm{V}_1)^*_{jk_1} - i (\bm{M}_1)^*_{jk_1}  ) \hat{d}_{k_1}^\dag \kett{\rho_A(0)}, \\
        e^{\hat{\mathcal{L}}_A t} \hat{\mathcal{F}}_j^- \kett{\rho_A(0)} = \sum_{k_2 \in \mathbb{I}_2}  e^{(iE_{k_2}-\gamma_{k_2})t} ((\bm{V}_2)_{jk_2} + i (\bm{M}_2)_{jk_2}  )  \hat{d}_{k_2}  \kett{\rho_A(0)}.
    \end{gathered}
    \label{smeq:impurity_directcalculation1}
\end{equation}
From the conjugation rule in Eq.~\ref{smeq:conjugation}, $\braa{I} \hat{d}_k^\dag = \braa{I} \hat{\Tilde{d}}_k^\dag$ and $\braa{I} \hat{d}_k = - \braa{I} \hat{\Tilde{d}}_k$. Therefore,
\begin{equation}\label{smeq:impurity_directcalculation2}
    \begin{gathered}
        \braa{I} \hat{\mathcal{F}}_j^- = \braa{I} \left( \hat{A}_j - \frac{i}{2}  \hat{L}_{1j} - \frac{i}{2} \hat{L}_{2j} \right) =  \sum_{k \in \mathbb{I}_1 \cup \mathbb{I}_2} \braa{I} (V_{jk}-iM_{jk}) \hat{d}_k, \\
        \braa{I} \hat{\mathcal{F}}_j^+ = \braa{I} \left( \hat{A}_j^\dag + \frac{i}{2} \hat{L}^\dag_{1j} + \frac{i}{2}  \hat{L}^\dag_{2j} \right) =  \sum_{k \in \mathbb{I}_1 \cup \mathbb{I}_2} \braa{I} (V_{jk}-iM_{jk})^* \hat{d}^\dag_k.
    \end{gathered}
\end{equation}
Combining Eq.~\ref{smeq:impurity_directcalculation1} and Eq.~\ref{smeq:impurity_directcalculation2}, the BCFs are expressed as,
\begin{equation}
    \begin{gathered}
        \bm{C}^{>A}(t+t',t') = (\bm{V}_{1}-i\bm{M}_{1}) e^{-\bm{Z}_{1}t} (\bm{V}_{1}+i\bm{M}_{1})^\dag, \\
        \bm{C}^{<A}(t+t',t') = (\bm{V}_{2}-i\bm{M}_{2})^* e^{-\bm{Z}_{2} t} (\bm{V}_{2}+i\bm{M}_{2})^T, 
    \end{gathered}
\end{equation}

The equivalence condition with the BCF in the fermionic bath can be similarly constructed with the Dyson series expansion as in the bosonic bath. For a rigorous construction of the fermionic influence superoperator, we refer to Ref.~\cite{Cirio2022}.

\section{Optimal gauge choice minimizing violation of CP condition}

In the single coupling limit with a diagonal $\bm{Z}$, the relationship between the weight $w$ and the coupling parameters $V$ and $M$ is established as $(V - i M)(V^* - i M^*) = w$. Here, $w, V, M \in \mathbb{C}$. Once $w$ is given, the coupling parameters $V$ and $M$ have a gauge choice parameterized as $V - i M = \kappa \sqrt{w}, V + i M = \kappa^{*-1} \sqrt{w^*}$ where $\kappa \in \mathbb{C}$. We show that the solution of $V$ and $M$ with minimum $|M|$ is then given by $V - iM = \sqrt{w}$ with $V, M \in \mathbb{R}$, ($\kappa = 1$), or $V = \operatorname{Re}( \sqrt{w}), \: M = -\operatorname{Im}(\sqrt{w})$.
\begin{prop}
    Given $w \in \mathbb{C}$, for $\forall \ V, M \in \mathbb{C}$ satisfying the constraint, $(V - i M)(V^* - i M^*) = w$, and its real solution, $V_r, M_r \in \mathbb{R}$, $V_r = \operatorname{Re}(\sqrt{w}), M_r = - \operatorname{Im}(\sqrt{w})$, then $|M| \geq |M_r|$.
\end{prop}
\begin{proof}
    With a parameterization, $V = |V|e^{i \theta_1}$, $M= |M|e^{i \theta_2}$, the constraint becomes, $|V|^2 - |M|^2 - 2i|V||M|\cos(\theta_1 - \theta_2) = w$ or $|V|^2 - |M|^2 = \operatorname{Re}(w)$ and $|V||M|\cos(\theta_1 - \theta_2) = -\operatorname{Im}(w)/2$. It leads to 
    $$
    |M| \sqrt{\operatorname{Re}(w) + |M|^2} |\cos(\theta_1 - \theta_2)| = |\operatorname{Im}(w)|/2. 
    $$
    Since $|M|\sqrt{\operatorname{Re}(w) + |M|^2}$ is an increasing function in $|M|$, we have minimum $|M|$ when $|\cos (\theta_1 - \theta_2)|=1$, which corresponds to $\theta_1 = \theta_2$ or $\theta_1 = \theta_2 \pm \pi$. The real solution $V_r, M_r$ satisfies this condition on $\theta_1$ and $\theta_2$, so $|M| \geq |M_r|$.
\end{proof}

\section{Proofs of the stability condition in the noninteracting fermionic impurity model}

Recall that for a noninteracting fermionic impurity model, its one-particle reduced density matrix (1-RDM) $P_{pq}(t) = \tr[\hat{c}_q^\dag \hat{c}_p \hat{\rho}_{SA}(t)]$ satisfies a continuous differential Lyapunov equation~\cite{Lotem2020, Barthel_2022}:
\begin{equation}
    \partial_t \bm{P} = \bm{X}\bm{P} + \bm{P}\bm{X}^\dag + \bm{Y},
    \label{smeq:lyapunov_rdm}
\end{equation}
where
\begin{equation}
    \begin{gathered}
        \bm{X} = \left(
        \begin{array}{ccc}
           -i \boldsymbol{H}_{S}  & -i\bm{V}_1-\bm{M}_1 &  -i\bm{V}_2-\bm{M}_2 \\
           -i\bm{V}_1^\dag - \bm{M}_1^\dag  & -\bm{Z}_1 & \bm{0} \\
            -i\bm{V}_2^\dag - \bm{M}_2^\dag & \bm{0} & -\bm{Z}_2
        \end{array}
        \right), \\
        \bm{Y} = \left(
        \begin{array}{ccc}
           \bm{0}  & \bm{0} &  2\bm{M}_2 \\
           \bm{0}  & \bm{0} & \bm{0} \\
            2\bm{M}_2^\dag & \bm{0} & 2 \bm{\Gamma}_2
        \end{array}
    \right).
    \end{gathered}
\end{equation}

\begin{prop} 
Eq.~\ref{smeq:lyapunov_rdm} is asymptotically stable if and only if the real part of all eigenvalues of $\bm{X}$ is strictly negative.
\label{lemma:depend_on_X}
\end{prop}
\begin{proof}
    Let $\text{Vec}(\bm{P})$ denote the vectorization of $\bm{P}$. Then Eq.~\ref{smeq:lyapunov_rdm} can be rewritten as:
    \begin{equation}
        \partial_t\text{Vec}(\bm{P}) = \left(\bm{X}\otimes \bm{I} + \bm{I}\otimes \bm{X}^* \right)\text{Vec}(\bm{P}) + \operatorname{Vec}(\bm{Y}).
    \end{equation}
    Therefore, $\bm{P}(t)$ is asymptotically stable if and only if all eigenvalues of $\bm{X}\otimes \bm{I} + \bm{I}\otimes \bm{X}^*$ have strictly negative real parts. Note that $(\lambda, \bm{u})$ is an eigen-pair of $\bm{X}\otimes \bm{I}$ if and only if $\bm{u} = u_1\otimes u_2^*$ and $\bm{X}u_1=\lambda u_1$, and similarly for $\bm{I} \otimes \bm{X}^*$. Also, note that $\bm{X} \otimes \bm{I}$ and $\bm{I} \otimes \bm{X}^*$ commutes. 
    Therefore $\widetilde\lambda_{ij}=\lambda_i+\lambda_j^*$ is the eigenvalue of $\bm{X} \otimes \bm{I}+\bm{I} \otimes \bm{X}^*$ where $\lambda_i,\lambda_j$ are the eigenvalues of $\bm{X}$. 
    Therefore, Eq.~\ref{smeq:lyapunov_rdm} is asymptotically stable if and only if all eigenvalues of $\bm{X}$ have strictly negative real parts. 
\end{proof}

For simplicity, we now assume a single-site system, i.e., $\hat{H}_S = h \hat{a}^\dag \hat{a}$ where $h \in\mathbb R$, and $\bm{X}$ is
\begin{equation}
    \bm{X} = \left( 
    \begin{array}{cc}
         -i h &  -i\bm{V}-\bm{M} \\
         -i\bm{V}^\dag - \bm{M}^\dag & -\bm{Z}
    \end{array}
    \right),
    \label{smeq:single_site_X}
\end{equation}
where we do not separate baths 1 and 2 since there is no difference in the expression for $\bm{X}$. The matrices $\bm{V}$ and $\bm{M}$ are now row vectors of size $1 \times N_b$, where $N_b$ is the total number of bath modes. 
When needed, we will interchangeably refer to these matrices as vectors.

\begin{prop}
    When $\|\bm{M}\|<\frac{1}{2} \min \gamma_{k}$ for $\forall k$, there exists $m>0$, such that if $\|\bm{V}\|>m$, the dynamics Eq.~\ref{smeq:lyapunov_rdm} with $\bm{X}$ being Eq.~\ref{smeq:single_site_X} is asymptotically stable.
    \label{smtheorem:stability}
\end{prop}
\begin{proof}
With Proposition~\ref{lemma:depend_on_X}, Eq.~\ref{smeq:lyapunov_rdm} being asymptotically stable is equivalent to $
\partial_t u = \bm{X}u$
being asymptotically stable for any initial vector $u(0)$. To study the stability of $\partial_t u = \bm{X}u$, let us define a Lyapunov function $\mathcal{V}(u)= \frac{1}{2}u^{\dagger} \bm{\Theta} u$, for which the quadratic form $\bm{\Theta}$ will be specified later.
Then we have
$$
\dot{\mathcal{V}}(u(t)) = \frac 1 2 u^{\dagger}(\bm{\Theta}\bm{X} + \bm{X}^{\dagger} \bm{\Theta})u,
$$
Then, by the Lyapunov stability theorem, if $\bm{\Theta}\bm{X} + \bm{X}^{\dagger} \bm{\Theta}$ is negative definite, then $\partial_t u = \bm{X}u$ is asymptotically stable.
Let us consider
\begin{equation}
    \bm{\Theta} = \left(
    \begin{array}{cc}
        1 & -i \frac{\epsilon \boldsymbol{V} }{1 + |\boldsymbol{V}|^2}  \\
        i \frac{\epsilon \boldsymbol{V}^\dag }{1 + |\boldsymbol{V}|^2 } & \boldsymbol{I}
    \end{array}
    \right).
\end{equation}
Here $\epsilon$ are parameters that will be determined later. We define $Q(\boldsymbol{V}) = - (\bm{\Theta}\bm{X} + \bm{X}^{\dagger} \bm{\Theta})$. Let $\boldsymbol{V} = c \widetilde{\boldsymbol{V}}$ where $\widetilde{\boldsymbol{V}}$ is a unit vector, then we have
$$ \lim_{c\rightarrow \infty} Q(c \widetilde{\boldsymbol{V}}) = -\lim_{c\rightarrow \infty}(\bm{\Theta}\bm{X} + \bm{X}^{\dagger} \bm{\Theta}) = \left(\begin{array}{cc}
  2\epsilon   & 2 \boldsymbol{M}\\
  2\boldsymbol{M}^\dag   & 2\boldsymbol{\Gamma} - 2\epsilon\widetilde{\boldsymbol{V}}^\dag \widetilde{\boldsymbol{V}} 
 \end{array}\right) := Q_{\infty}.
 $$
If we let $\epsilon = \frac{1}{2} {\min \gamma_k}$, $\boldsymbol{\Gamma}\geq 2\epsilon \boldsymbol{I}$, and $Q_{\infty} \geq \left(
 \begin{array}{cc}
  2\epsilon    & 2\boldsymbol{M} \\
  2\boldsymbol{M}^\dag    & 4\epsilon \boldsymbol{I} -2\epsilon\widetilde{\boldsymbol{V}}^\dag \widetilde{\boldsymbol{V}}
 \end{array}
 \right)\geq \left(
 \begin{array}{cc}
  2\epsilon    & 2\boldsymbol{M} \\
  2\boldsymbol{M}^\dag    & 2\epsilon \boldsymbol{I}
 \end{array}
 \right)$. Since $\|\boldsymbol{M}\| <\epsilon= \frac{1}{2}\min\gamma_k$, then $Q_{\infty}>0$. Therefore, for any unit vector $\widetilde{\boldsymbol{V}}_0$, there exists a constant $c_0>0$ depending on $\widetilde{\boldsymbol{V}}_0$ such that if $c>c_0$, then $Q = Q(c\widetilde{\boldsymbol{V}}_0)> 0$. Since both $Q$ and $c_0$ depend continuously on $\widetilde{\boldsymbol{V}}$, and the unit sphere is compact, there exists a positive constant $m$ such that when
 $c=\|\boldsymbol{V}\|>m$, we have $Q>0$ and the dynamics is asymptotically stable.
\end{proof}
Proposition \ref{smtheorem:stability} implies that the dynamics, though violating the positivity condition, is stable with sufficiently large system-bath Hamiltonian coupling $\hat H_{SA}$ and small Lindblad coupling $\mathcal{D}_{SA}$.

\section{Details of Fitting procedures}

\subsection{ESPRIT algorithm for complex exponential function fitting with complex weights}

In this section, we briefly describe the Estimation of Signal Parameters via Rotational Invariant Techniques  (ESPRIT) algorithm~\cite{PaulrajRoyKailath1986} used for fitting the BCF as a sum of complex exponential functions with complex weights, $ C(t) \approx \sum_k w_k  e^{- z_k t }$, where $w_k$ and $z_k$ are complex-valued. The  ESPRIT algorithm has been widely adopted in a signal processing community~\cite{RoyKailath1989, RouquetteNajim2001, ShahbazpanahiValaeeBastani2001, Fannjiang2020, DingEpperlyLinetal2024}, and has recently been applied to quantum computing applications~\cite{StroeksTerhal2022, LiNiYing2023}. In this work, we take the ESPRIT algorithm based on the following Hankel matrix~\cite{Fannjiang2020}:
\begin{equation}
    H = \left[\begin{array}{ccccc}
    C(t_0) & C(t_1) & C(t_2) & \cdots & {C(t_{n-1})} \\
    C(t_1) & C(t_2) & C(t_3) & \cdots & C(t_n) \\
    C(t_2) & C(t_3) & C(t_4) & \cdots & {C(t_{n+1})} \\
    \vdots & \vdots & \vdots & \ddots & \vdots \\
    C(t_{n-1}) & C(t_{n}) & C(t_{n+1}) & \cdots & C(t_{2n-2})
    \end{array}\right]
\end{equation}
or $H_{ab} = C(t_{a+b})$, where $t_a = a\delta t$ after discretizing time with a timestep $\delta t$. 
The ESPRIT algorithm is summarized in the following steps to obtain the complex exponent $z_k$.
\begin{itemize}
    \item Step 1: Perform the singular-value decomposition (SVD) of the matrix $H$: 
    $H = U\Sigma V$, where the singular values (diagonal entries of the diagonal matrix $\Sigma$) are ordered non-increasingly.
    \item Step 2: let $U_1 = U[1:n-1, 1:N_{\exp}]$, $U_2 = U[2:n,1:N_{\exp}]$. Calculate $W = U_1^{+} U_2$, where $U_1^{+} = (U_1^{\dagger}U_1)^{-1}U_1^{\dagger}$ is the Moore-Penrose pseudo inverse of $U_1$. 
    \item Step 3: Calculate the eigenvalue $\lambda_1,\cdots,\lambda_{N_{\exp}}$ of the matrix $W$.
    Let $z_k = -\log(\lambda_k)/\Delta t$, where we select the branch of $\log(\cdot)$ such that $\operatorname{Im}(\log(\lambda_k))\in 
 [-\pi,\pi)$.
\end{itemize}
After we get the complex exponent $z_k$, the complex weights $w_k$ are obtained via the following least-squares fitting:
\begin{equation}\label{smeq:leastsquare}
    \min_{\{w_k\}} \sum_a \left|C(t_a)
   -\sum_k w_k  e^{-z_k  t_a}
    \right|^2.
\end{equation}
In the multi-site case, the BCFs are given by a matrix-valued function, $C_{jj'}(t) \approx \sum_k (\bm{W}_k)_{jj'} e^{-z_k t}$. We construct a scalar quantity $C(t)$ by summing over all the entries in $C_{jj'}(t)$, i.e., $C(t) = \sum_{jj'} C_{jj'}(t) $, to estimate $z_k$ as above. Then, we obtain the matrix-valued complex weights, $(\bm{W}_k)_{jj'}$, with the elementwise least-squares fitting with Eq.~\ref{smeq:leastsquare}.

\subsection{Complex exponential function fitting with positive weights}
As a comparison to the complex exponential function fitting with complex weights, we performed the complex exponential function fitting with positive weights (i.e., $w_k \in \mathbb{R}^+$), which is a BCF decomposition for the Lorentzian pseudomode model, with numerical fitting using the L-BFGS-B algorithm. We minimized the following fitting error function,
\begin{equation}
    \mathcal{E}_{\text{fit}} = \delta t \sqrt{ \sum_a \left|C(t_a)
   -\sum_k w_k  e^{-z_k  t_a}
    \right|^2} .
\end{equation}
During the optimization steps, we only used the gradient with respect to $z_k$ and determined $w_k$ from the non-negative least-squares fitting following \cite{HuangGullLin2023}. 
As an initial guess for the optimization problem, we followed the deterministic procedure from Ref.~\cite{Lotem2020, Brenes2020, Zwolak2017, Trivedi2021}, which determines the imaginary part of $z_k$ from the direct discretization of the unitary bath and sets the real part of $z_k$ to some finite values. We chose the imaginary part of $z_k$ from a set of Gauss-Legendre quadrature nodes on the real axis, and the real part of $z_k$ as $\text{Re}[z_k] = 0.1 \times W/N_{\exp}$ where $W$ is a width of the spectrum. 

In the multi-site case, the fitting error function is modified as,
\begin{equation}
    \mathcal{E}_{\text{fit}} = \delta t \sqrt{ \sum_{jj'} \sum_a \left|C_{jj'}(t_a)
   -\sum_k (\bm{W}_k)_{jj'} e^{-z_k  t_a}
    \right|^2},
\end{equation}
with the constraint that $\bm{W}_k\ge 0$.

\section{Details on DMRG calculations}

We describe some details of the time-dependent density matrix renormalization group (tdDMRG) calculations, focusing on the fermionic impurity models. We first express the fermionic Lindblad master equation in the super-fermion representation. This allows us to represent the density operator as a state vector with Schr\"{o}dinger equation-like propagation so that the two-site time-dependent variational principle, implemented in the Block2~\cite{Zhai2021, Zhai2023} package, can be directly used for the state vector propagation. Furthermore, the Liouville operator has a number-conserving form in the super-fermion representation, so we can take advantage of $\mathbb{U}(1)$ symmetry adaptation in the tdDMRG calculations.

Ordering of orbitals is important in DMRG calculations. We employed two different ordering schemes illustrated in Fig.~\ref{smfig:dmrg}. The illustration is based on the spinless single-site impurity model with two bath orbitals in each bath $A_1$ and $A_2$. We have $\hat{\rho}_{SA}(0) = \hat{\rho}_{S}(0) \otimes \hat{\rho}_{A_1}(0) \otimes \hat{\rho}_{A_2}(0)$, where $\hat{\rho}_{A_1}(0) = \proj{00}$ and $\hat{\rho}_{A_2}(0) = \proj{11}$. In the super-fermion representation, $\kett{\rho_{SA}(0)} = \kett{\rho_S(0)} \otimes \kett{\rho_{A_1}(0)} \otimes \kett{\rho_{A_2}(0)}$ where $\kett{\rho_{A_1}(0)} = \kett{0101}$ and $\kett{\rho_{A_2}(0)}=\kett{1010}$. The orbitals from the `tilde' space are drawn with hatching lines in Fig.~\ref{smfig:dmrg}. The two orbital ordering schemes are distinguished by whether the bath $A_1$ and $A_2$ are on the same side (Fig.~\ref{smfig:dmrg}a) or different sides (Fig.~\ref{smfig:dmrg}b) of the impurity. We chose the first scheme for the spinful impurity model to put two different spins on different sides of the impurity~\cite{Saberi2008}. We chose the second scheme for the spinless impurity model. The remaining ordering of the orbitals within each bath is determined from the imaginary part of $z_k$, or equivalently, the energy part, $E_k$, following the physical insights in \cite{ RamsZwolak2020, Zwolak2020}.

\begin{figure}[ht]
    \centering
    \includegraphics[width=0.7\columnwidth]{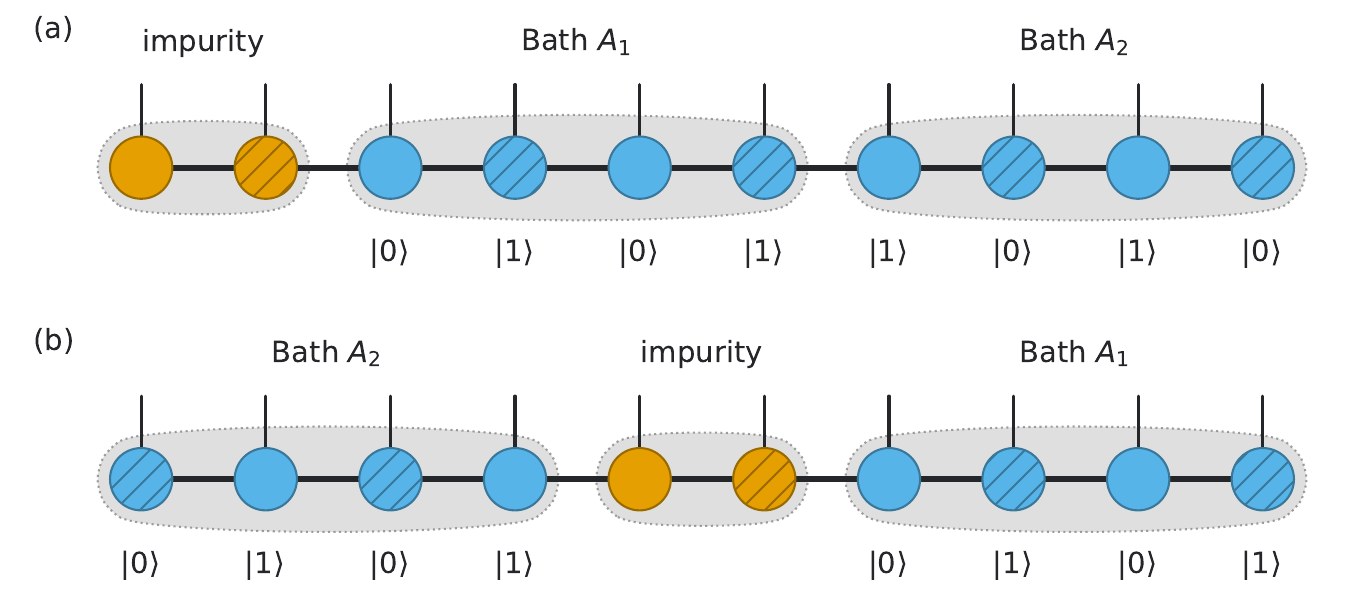}
    \caption{Two different orbital ordering schemes for tdDMRG calculations.}
    \label{smfig:dmrg}
\end{figure}